\documentclass[conference]{IEEEtran}

\usepackage{algorithm}
\usepackage[noend]{algpseudocode}

\usepackage{amsmath}
\usepackage{amssymb}

\usepackage{amsthm}
\usepackage{color}
\usepackage{graphicx}
\usepackage{listings}
\usepackage{rotating}
\usepackage{subfig}
\usepackage{tikz}
\usepackage{xspace}
\usepackage{colortbl}
\usepackage{graphicx}
\usepackage{balance}
\usepackage{url}
\usepackage[colorlinks=true,
            urlcolor=blue,
            citecolor=blue,
            linkcolor=blue
            ]{hyperref}

\usepackage{mathtools}
\usepackage{paralist}

\usetikzlibrary{decorations,arrows,shapes,shadows}

\usepackage[disable]{todonotes}

\newcommand{\BG}[1]{\todo[inline,color=shadered]{\textbf{Boris says:} #1}}

\newcommand{\SL}[1]{\todo[inline,color=shadeblue]{\textbf{Seokki says:} #1}}

\newcommand{\SLDel}[1]{\todo[inline,color=shadeblue]{\textbf{Seokki deleted:} #1}}

\newcommand{\mypara}[1]{\noindent\textbf{#1.}}
\newcommand{\mypartitle}[1]{\smallskip\noindent\textbf{#1.}}

\newcommand{\listconcat}{\,{\tt ::}\,}

\newcommand{\dlImp}[0]{\,\ensuremath{\mathtt{{:}-}}\,}
\newcommand{\dlNeg}{\neg\,}
\newcommand{\bodyOf}[1]{body(#1)}
\newcommand{\headOf}[1]{head(#1)}
\newcommand{\varsOf}[1]{vars(#1)}
\newcommand{\attrsOf}[1]{attrs(#1)}
\newcommand{\argsOf}[1]{args(#1)}
\newcommand{\predOf}[1]{pred(#1)}
\newcommand{\matches}{\curlyeqprec}
\newcommand{\rel}[1]{\ensuremath{\mathtt{#1}}}

\newcommand{\depthP}[1]{d(#1)}
\newcommand{\aDepth}{d}
\newcommand{\isSucc}{\models}
\newcommand{\isFailed}{\not\models}

\usepackage[yyyymmdd,hhmmss]{datetime}

\newcommand{\KPlus}{\ensuremath{+}\xspace}
\newcommand{\KTimes}{\ensuremath{\times}\xspace}
\newcommand{\KTimesB}{\ensuremath{\bullet}\xspace}
\newcommand{\KNeg}{\ensuremath{\neg}\xspace}

\newcommand{\ProvPoly}{\ensuremath{\mathbb{N}[X]}\xspace}

\newcommand{\provGraph}{\ensuremath{{\cal PG}}\xspace}

\newcommand{\nodeLabel}{{\cal L}}
\newcommand{\successLabel}{{\cal S}}
\newcommand{\stringDom}{\mathbb{L}}

\newcommand{\GPProg}{\mathbb{GP}}
\newcommand{\whyq}{\textsc{Why}\,}
\newcommand{\whynotq}{\textsc{Whynot}\,}
\newcommand{\explainq}{\textsc{Expl}}

\newcommand{\qType}{typeof}

\newcommand{\provQ}{\textsc{PQ}}

\newcommand{\unProg}[1]{#1_U}
\newcommand{\adProg}[1]{#1_A}
\newcommand{\fireProg}[1]{{#1}_{Fire}}
\newcommand{\fireCProg}[1]{#1_{FC}}
\newcommand{\moveProg}[1]{#1_M}
\newcommand{\fire}[3]{\rel{F_{#1,#3}}}
\newcommand{\fireC}[4]{\rel{FC_{#1,{#4},{#3}}}}
\newcommand{\boolT}{true}
\newcommand{\boolF}{false}
\newcommand{\adornment}{\sigma}
\newcommand{\greenT}{\textcolor{DarkGreen}{T}}
\newcommand{\redF}{\textcolor{DarkRed}{F}}

\newcommand{\nodeSk}[3]{f_{#2}^{#3}}

\definecolor{DarkGreen}{rgb}{0,0.45,0}
\definecolor{DarkRed}{rgb}{0.8,0,0}
\definecolor{DarkYellow}{rgb}{0.6,0.6,0}
\definecolor{DarkGray}{rgb}{0.2,0.2,0.2}

\newcommand{\adom}[1]{\ensuremath{\mathit{adom}(#1)}}
\newcommand{\domA}{\ensuremath{\mathit{dom}}}

\newcommand{\thead}[1]{{\cellcolor{black}{\textcolor{white}{\textbf{#1}}}}}

\usepackage{multirow}

\newcommand{\mathtab}{\ensuremath\thickspace\thickspace\thickspace}

\newcommand{\card}[1]{\| {#1} \|}
\algrenewcommand\algorithmicindent{1.1em}

\newtheorem{Theorem}{Theorem}
\newtheorem{Definition}{Definition}

\newtheorem{Example}{Example}

\definecolor{black}{rgb}{0,0,0}
\definecolor{grey}{rgb}{0.8,0.8,0.8}
\definecolor{red}{rgb}{1,0,0}
\definecolor{green}{rgb}{0,1,0}
\definecolor{darkgreen}{rgb}{0,0.5,0}
\definecolor{darkpurple}{rgb}{0.5,0,0.5}
\definecolor{darkdarkpurple}{rgb}{0.3,0,0.3}
\definecolor{blue}{rgb}{0,0,1}
\definecolor{shadegreen}{rgb}{0.95,1,0.95}
\definecolor{shadeblue}{rgb}{0.95,0.95,1}
\definecolor{shadered}{rgb}{1,0.85,0.85}
\definecolor{shadegrey}{rgb}{0.85,0.85,0.85}
\definecolor{oddRowGrey}{rgb}{0.80,0.80,0.80}
\definecolor{evenRowGrey}{rgb}{0.85,0.85,0.85}

\input{prolog_sty}

\title{\bf Efficiently Computing Provenance Graphs for Queries with Negation}

\newcommand{\emailaddr}[1]{\texttt{\small #1}}

\author{
 \IEEEauthorblockN{
 Seokki Lee\IEEEauthorrefmark{1}
 \hfill Sven K\"ohler\IEEEauthorrefmark{2}
 \hfill Bertram Lud\"ascher\IEEEauthorrefmark{3}
 \hfill Boris Glavic\IEEEauthorrefmark{1}
 }
 \IEEEauthorblockA{\IEEEauthorrefmark{1}Illinois Institute of
   Technology. ~ \emailaddr{\{slee195@hawk.iit.edu}, 
 		\emailaddr{bglavic@iit.edu\}}}
 \IEEEauthorblockA{\IEEEauthorrefmark{3}University of Illinois at
   Urbana-Champaign. ~ \emailaddr{\{ludaesch@illinois.edu\}}}
 \IEEEauthorblockA{\IEEEauthorrefmark{2}University of California at
   Davis. ~ \emailaddr{\{svkoehler@ucdavis.edu\}}}
}

\newcommand{\TRtitle}{Efficiently Computing Provenance Graphs for Queries with Negation}
\newcommand{\TRauthors}{Seokki Lee, Sven K\"ohler, Bertram Lud\"ascher, Boris Glavic}
\newcommand{\TRnumber}{IIT/CS-DB-2016-03}
\newcommand{\TRdate}{2016-10}
\begin{document}

\twocolumn[{

\vspace{1cm}

\begin{minipage}{0.5\linewidth}
  \colorbox{black}{\includegraphics[width=1\linewidth]{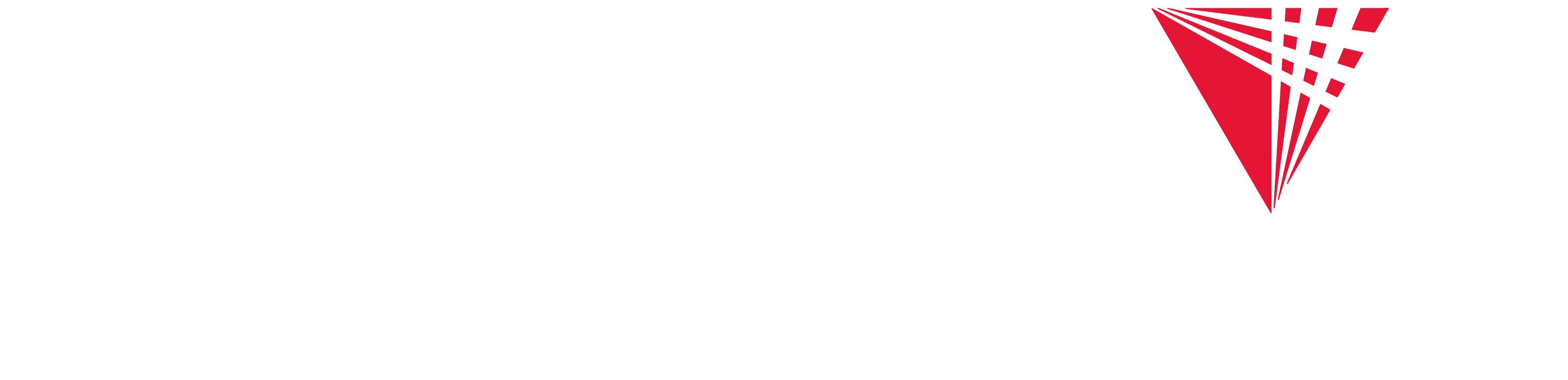}}
\end{minipage}
\hfill
\begin{minipage}{0.16\linewidth}
  \includegraphics[width=1\linewidth]{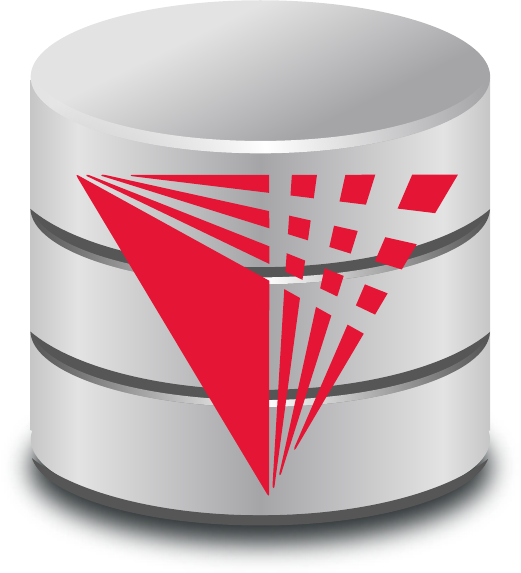}  
\end{minipage}\\
\vspace{4cm}

\centering
\begin{minipage}{1.0\linewidth} 
\centering
{\Huge \bf \TRtitle}
\end{minipage}
\\
\vspace{1cm}

{\huge \TRauthors}\\
\vspace{1cm}

{\huge \tt IIT DB Group Technical Report \TRnumber}\\
\vspace{1cm}

{\Large \TRdate}
\vspace{3cm}

{\huge \url{http://www.cs.iit.edu/~dbgroup/}}

\vspace{3cm}
\begin{minipage}{1.0\linewidth}
\textbf{LIMITED DISTRIBUTION NOTICE}: The research presented in this report may be submitted as a whole or in parts for publication  and will probably be copyrighted if accepted for publication. It has been issued as a Technical Report for early dissemination of its contents. In view of the transfer of copyright to the outside publisher, its distribution outside of IIT-DB prior to publication should be limited to peer communications and specific requests. After outside publication, requests should be filled only by reprints or legally obtained copies of the article (e.g. payment of royalties).  
\end{minipage}

}]

\clearpage

\graphicspath{ {./figures/} }
\maketitle

\begin{abstract}
Explaining why an answer is in the result of a query or why it is missing from the result is important for many applications including auditing, debugging data and queries, and answering hypothetical questions about data.
Both types of questions, i.e., \emph{why} and \emph{why-not} 
provenance, 
have been studied extensively.
In this work, we present the first \emph{practical} approach for 
answering such questions for queries with negation (first-order queries). 
Our approach is based on a rewriting of Datalog rules (called \emph{firing rules}) that captures successful rule derivations within the context of a Datalog query.
We extend this rewriting 
to support negation and to capture failed derivations that explain missing answers. 
Given a (why or why-not) provenance question, we compute an \emph{explanation}, i.e., the part of the provenance 
that is relevant to answer the question. We introduce optimizations that prune parts of a provenance graph early on if we can determine that they will not be part of the explanation for a given question.
We present an implementation that runs on top of a relational database using SQL to compute explanations. 
Our experiments demonstrate that our approach scales to large instances and significantly outperforms an earlier approach which instantiates the full provenance to compute explanations. 
\end{abstract}




\section{Introduction}
\label{sec:intro}

Provenance for relational queries records how results of a query
depend on the query's inputs.  This type of information can be used to
explain \emph{why} (and \emph{how}) a result is derived by a query
over a given database.  Recently, approaches have been developed that
use provenance-like techniques to explain why a tuple (or a set of tuples described declaratively by a 
pattern) is \emph{missing} from the query
result.  However, the two problems of computing provenance and
explaining missing answers have been treated mostly in isolation. A notable exception is~\cite{MG10} which computes causes for answers and non-answers. 
However, the approach requires the user to specify which missing inputs to consider as causes for a missing output.
Capturing provenance for a query with negation necessitates the
unification of \emph{why} and \emph{why-not} provenance, because to explain a result of the query we have to describe 
how existing and missing intermediate results (via positive and
negative subqueries, respectively) lead to the creation of the result. This has also been recognized by K\"ohler et al.~\cite{KL13}: asking
why a tuple $t$ is absent from the result of a query $Q$ is equivalent to asking why $t$ 
is present in $\neg Q$. Thus, a provenance model that supports queries
with negation naturally supports why-not provenance. 
In this paper, 
we present a framework 
that answers 
why and why-not 
questions for
queries with 
negation.  
To this end, we introduce a graph model 
for 
provenance of
first-order (FO) 
queries (i.e., non-recursive Datalog 
with negation) and an efficient method for explaining a (missing) answer using SQL. 
Our approach is based on the observation that typically only a part of provenance, 
which we call \textit{explanation} in this work, 
is actually relevant for answering the user's provenance question 
about the existence or absence of a 
result.

%
%
%
%
%
\begin{figure}[t]
  \centering $\,$\\[-3mm]
    \begin{minipage}{1\linewidth}
      \centering
      \begin{align*}
		&r_1: \rel{Q}(X,Y) :- \rel{Train}(X,Z), \rel{Train}(Z,Y), \neg \rel{Train}(X,Y)\\
      \end{align*}
    \end{minipage}\\[-4mm]
    \begin{minipage}{1.1\linewidth}
      \begin{minipage}{1\linewidth}
     \begin{minipage}{0.54\linewidth}
     \begin{minipage}{0.2\linewidth}
       \resizebox{1\columnwidth}{!}{\begin{minipage}{1.4\linewidth}
\begin{tikzpicture}[>=latex',
line join=bevel,
line width=0.4mm,
every node/.style={ellipse},
minimum height=4mm]

  \definecolor{fillcolor}{rgb}{0.0,0.0,0.0};
%
  \node (s) at (30bp,30bp) [draw=black,circle] {$s$};
  \node (c) at (60bp,30bp) [draw=black,circle] {$c$};
  \node (n) at (60bp,0bp) [draw=black,circle] {$n$};
  \node (w) at (30bp,0bp) [draw=black,circle] {$w$};

   \path[]
%
%
   		(n) edge [->] node [above] {} (w)
   		(w) edge [->]  node [left] {} (s)
   		(n) edge [->] node [below] {} (c)
   		(c) edge [in=40,out=120,->] node [right] {} (s)
   		(s) edge [->] node [above] {} (c);

\end{tikzpicture}
\end{minipage}

}  
     \end{minipage}
      \begin{minipage}{0.78\linewidth}
        \scriptsize \centering $\,$\\[-1mm]
        \begin{tabular}{|cc|}
          \multicolumn{2}{c}{Relation \textbf{Train}}  \\[0.5mm]\cline{1-2}
	  \thead {fromCity} & \thead {toCity} \\ 
          new york & washington dc \\
          new york & chicago \\
          chicago & seattle \\
	  seattle & chicago \\
          washington dc & seattle \\
          \cline{1-2}
        \end{tabular}
      \end{minipage}
    \end{minipage}
     \begin{minipage}{0.45\linewidth}

      \begin{minipage}{0.55\linewidth}
 	\scriptsize \centering
         \begin{tabular}{|cc|}
           \multicolumn{2}{c}{Result of query $\textbf{Q}$}  \\[0.5mm]\cline{1-2}
 	  \thead {X} & \thead {Y} \\ 
 	   washington dc & chicago\\ 
 	   new york & seattle \\ 
	   seattle & seattle \\
	   chicago & chicago \\
           \cline{1-2}
         \end{tabular}
     \end{minipage}\\
	
%
%
   \end{minipage}
      \end{minipage}\\ %
    \end{minipage}
    $\,$\\[-1.5mm]
  \caption{Example train connection database and query}
  \label{fig:running-example-db}
\end{figure}
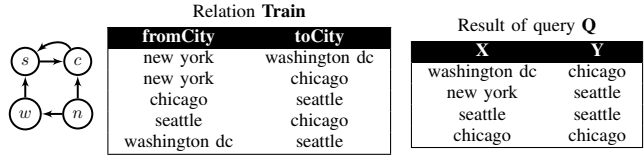
\begin{figure}[t]
  \centering
  $\,$\\[-3mm]
  \resizebox{0.73\columnwidth}{!}{\begin{tikzpicture}[>=latex',line join=bevel,line width=0.3mm]
  \definecolor{fillcolor}{rgb}{0.8,1.0,0.8};
  \node (REL_Q_WON_n_s_) at (150bp,175bp) [draw=black,fill=fillcolor,ellipse] {$Q(n,s)$};

  \node (RULE_0_LOS_n_s_w_) at (91bp,150bp) [draw=black,fill=fillcolor,rectangle] {$r_1(n,s,w)$};
  \node (GOAL_0_0_WON_n_w_) at (33bp,125bp) [draw=black,fill=fillcolor,rounded corners=.15cm,inner sep=3pt] {$g_{1}^{1}(n,w)$};
  \node (EDB_T_LOS_n_w_) at (33bp,95bp) [draw=black,fill=fillcolor,ellipse] {$T(n,w)$};

  \node (GOAL_0_1_WON_w_s_) at (93bp,125bp) [draw=black,fill=fillcolor,rounded corners=.15cm,inner sep=3pt] {$g_{1}^{2}(w,s)$};
  \node (EDB_T_LOS_w_s_) at (93bp,95bp) [draw=black,fill=fillcolor,ellipse] {$T(w,s)$};

  \node (GOAL_0_2_WON_n_s_) at (150bp,125bp) [draw=black,fill=fillcolor,rounded corners=.15cm,inner sep=3pt] {$g_{1}^{3}(n,s)$};
  \definecolor{fillcolor}{rgb}{1.0,0.51,0.51};
  \node (EDB_T_LOS_n_s_) at (150bp,95bp) [draw=black,fill=fillcolor,ellipse] {$T(n,s)$};

  \definecolor{fillcolor}{rgb}{0.8,1.0,0.8};
  \node (RULE_0_LOS_n_s_c_) at (211bp,150bp) [draw=black,fill=fillcolor,rectangle] {$r_1(n,s,c)$};
  \node (GOAL_0_0_WON_n_c_) at (209bp,125bp) [draw=black,fill=fillcolor,rounded corners=.15cm,inner sep=3pt] {$g_{1}^{1}(n,c)$};
  \node (EDB_T_LOS_n_c_) at (209bp,95bp) [draw=black,fill=fillcolor,ellipse] {$T(n,c)$};

  \node (GOAL_0_1_WON_c_s_) at (265bp,125bp) [draw=black,fill=fillcolor,rounded corners=.15cm,inner sep=3pt] {$g_{1}^{2}(c,s)$};
  \node (EDB_T_LOS_c_s_) at (265bp,95bp) [draw=black,fill=fillcolor,ellipse] {$T(c,s)$};

  \draw [->] (RULE_0_LOS_n_s_c_) -> (GOAL_0_0_WON_n_c_);
  \draw [->] (REL_Q_WON_n_s_) -> (RULE_0_LOS_n_s_c_);
  \draw [->] (RULE_0_LOS_n_s_c_) -> (GOAL_0_1_WON_c_s_);

  \draw [->] (GOAL_0_0_WON_n_c_) -> (EDB_T_LOS_n_c_);

  \draw [->] (REL_Q_WON_n_s_) -> (RULE_0_LOS_n_s_w_);
  \draw [->] (RULE_0_LOS_n_s_w_) -> (GOAL_0_1_WON_w_s_);
  \draw [->] (RULE_0_LOS_n_s_c_) -> (GOAL_0_2_WON_n_s_);
  \draw [->] (RULE_0_LOS_n_s_w_) -> (GOAL_0_2_WON_n_s_);
  \draw [->] (GOAL_0_2_WON_n_s_) -> (EDB_T_LOS_n_s_);

  \draw [->] (GOAL_0_1_WON_c_s_) -> (EDB_T_LOS_c_s_);

  \draw [->] (GOAL_0_0_WON_n_w_) -> (EDB_T_LOS_n_w_);


  \draw [->] (GOAL_0_1_WON_w_s_) -> (EDB_T_LOS_w_s_);

  \draw [->] (RULE_0_LOS_n_s_w_) -> (GOAL_0_0_WON_n_w_);
\end{tikzpicture}
}
  $\,$\\[-1mm]
  \caption{Provenance graph explaining $\whyq \rel{Q}(n,s)$}
  \label{fig:exam-pg-why-NY-seattle}
\end{figure}
\begin{figure}[t]
  \centering
  \resizebox{0.87\columnwidth}{!}{\begin{tikzpicture}[>=latex',line join=bevel,line width=0.3mm]

  \definecolor{fillcolor}{rgb}{1.0,0.51,0.51};
  \node (REL_Q_LOS_c_s_) at (285bp,240bp) [draw=black,fill=fillcolor,ellipse] {$Q(s,n)$};

  \node (RULE_0_WON_c_s_w_) at (180bp,213bp) [draw=black,fill=fillcolor,rectangle] {$r_1(s,n,w)$};
  \node (GOAL_0_0_LOS_c_w_) at (155bp,183bp) [draw=black,fill=fillcolor,rounded corners=.15cm,inner sep=3pt] {$g_{1}^{1}(s,w)$};
  \node (REL_TR_LOS_c_w_) at (145bp,152bp) [draw=black,fill=fillcolor,ellipse] {$T(s,w)$};
  \node (GOAL_0_1_LOS_w_n_) at (205bp,183bp) [draw=black,fill=fillcolor,rounded corners=.15cm,inner sep=3pt] {$g_{1}^{2}(w,n)$};
  \node (REL_TR_LOS_w_n_) at (205bp,152bp) [draw=black,fill=fillcolor,ellipse] {$T(w,n)$};

  \node (RULE_0_WON_c_s_n_) at (255bp,213bp) [draw=black,fill=fillcolor,rectangle] {$r_1(s,n,c)$};

  \node (GOAL_0_0_LOS_c_n_) at (255bp,183bp) [draw=black,fill=fillcolor,rounded corners=.15cm,inner sep=3pt] {$g_{1}^{2}(c,n)$};
  \node (REL_TR_LOS_c_n_) at (265bp,152bp) [draw=black,fill=fillcolor,ellipse] {$T(c,n)$};

  \node (RULE_0_WON_c_s_c_) at (320bp,213bp) [draw=black,fill=fillcolor,rectangle] {$r_1(s,n,s)$};
  \node (GOAL_0_0_LOS_c_c_) at (300bp,183bp) [draw=black,fill=fillcolor,rounded corners=.15cm,inner sep=3pt] {$g_{1}^{1}(s,s)$};
  \node (REL_TR_LOS_c_c_) at (320bp,152bp) [draw=black,fill=fillcolor,ellipse] {$T(s,s)$};
  \node (GOAL_0_1_LOS_s_n_) at (345bp,183bp) [draw=black,fill=fillcolor,rounded corners=.15cm,inner sep=3pt] {$g_{1}^{2}(s,n)$};
  \node (REL_TR_LOS_s_n_) at (375bp,152bp) [draw=black,fill=fillcolor,ellipse] {$T(s,n)$};

  \node (RULE_0_WON_c_s_s_) at (390bp,213bp) [draw=black,fill=fillcolor,rectangle] {$r_1(s,n,n)$};
  \node (GOAL_0_0_LOS_s_s_) at (390bp,183bp) [draw=black,fill=fillcolor,rounded corners=.15cm,inner sep=3pt] {$g_{1}^{1}(s,n)$};
  \node (GOAL_0_1_LOS_n_n_) at (435bp,183bp) [draw=black,fill=fillcolor,rounded corners=.15cm,inner sep=3pt] {$g_{1}^{2}(n,n)$};
  \node (REL_TR_LOS_n_n_) at (435bp,152bp) [draw=black,fill=fillcolor,ellipse] {$T(n,n)$};

  \draw [->] (RULE_0_WON_c_s_s_) -> (GOAL_0_0_LOS_s_s_);
  \draw [->] (RULE_0_WON_c_s_s_) -> (GOAL_0_1_LOS_n_n_);
  \draw [->] (GOAL_0_1_LOS_n_n_) -> (REL_TR_LOS_n_n_);
  \draw [->] (REL_Q_LOS_c_s_) -> (RULE_0_WON_c_s_w_);
  \draw [->] (GOAL_0_0_LOS_c_w_) -> (REL_TR_LOS_c_w_);
  \draw [->] (GOAL_0_0_LOS_c_c_) -> (REL_TR_LOS_c_c_);
  \draw [->] (REL_Q_LOS_c_s_) -> (RULE_0_WON_c_s_s_);
  \draw [->] (RULE_0_WON_c_s_n_) -> (GOAL_0_0_LOS_c_n_);
  \draw [->] (RULE_0_WON_c_s_c_) -> (GOAL_0_0_LOS_c_c_);
  \draw [->] (RULE_0_WON_c_s_c_) -> (GOAL_0_1_LOS_s_n_);
  \draw [->] (GOAL_0_1_LOS_s_n_) -> (REL_TR_LOS_s_n_);
  \draw [->] (GOAL_0_0_LOS_s_s_) -> (REL_TR_LOS_s_n_);
  \draw [->] (REL_Q_LOS_c_s_) -> (RULE_0_WON_c_s_n_);
  \draw [->] (RULE_0_WON_c_s_w_) -> (GOAL_0_0_LOS_c_w_);
  \draw [->] (RULE_0_WON_c_s_w_) -> (GOAL_0_1_LOS_w_n_);
  \draw [->] (GOAL_0_1_LOS_w_n_) -> (REL_TR_LOS_w_n_);

  \draw [->] (REL_Q_LOS_c_s_) -> (RULE_0_WON_c_s_c_);
  \draw [->] (GOAL_0_0_LOS_c_n_) -> (REL_TR_LOS_c_n_);

\end{tikzpicture}
}
  $\,$\\[-1mm]
  \caption{Provenance graph explaning $\whynotq \rel{Q}(s,n)$}
  \label{fig:exam-pg-whynot-chicago-seattle}
\end{figure}

\begin{Example}
\label{ex:example1}
Consider the relation $\rel{Train}$ in Fig.\,\ref{fig:running-example-db} that stores 
train connections in the US.
The Datalog rule $r_1$ in Fig.\,\ref{fig:running-example-db} computes which cities can be reached 
with exactly one transfer, but not directly.
We use the following abbreviations in provenance graphs: T = Train;
n = New York; s = Seattle; w = Washington DC and c = Chicago.
Given the result of this query,  
the user may be interested to know why he/she is
able to reach Seattle from New York ($\whyq \rel{Q}(n,s)$) with one
intermediate stop but not directly 
or why it is not possible to reach New York from Seattle  
in the same fashion ($\whynotq \rel{Q}(s,n)$).
\end{Example}
%
An explanation for either type of question should justify the existence (absence) of a 
result 
as the success (failure) 
to derive the result 
through the rules of the query. 
Furthermore, it should explain how the existence (absence) of tuples in the database caused the derivation to succeed (fail). 
Provenance graphs 
providing this type of justification for $\whyq \rel{Q}(n,s)$ and $\whynotq \rel{Q}(s,n)$ are shown in 
Fig.\,\ref{fig:exam-pg-why-NY-seattle} and
Fig.\,\ref{fig:exam-pg-whynot-chicago-seattle}, respectively.
There are three types of graph nodes: 
\emph{rule nodes} (boxes labeled with a rule identifier and the constant arguments 
of a rule derivation), 
\emph{goal nodes} (rounded boxes labeled with a rule identifier and the goal's position in the rule's body), and \emph{tuple nodes} (ovals). 
In these provenance graphs, nodes are either colored as 
\emph{green} (successful/existing) or \emph{red} (failed/missing).
%

\begin{Example}
Consider the explanation (provenance graph in Fig.\,\ref{fig:exam-pg-why-NY-seattle}) for question $\whyq \rel{Q}(n,s)$.  
Seattle can be reached from New York by either stopping in
Washington DC or Chicago and there is no direct connection between these two cities. 
These two options correspond to two successful
derivations 
for 
rule $r_1$ with $X {=} n$, $Y {=} s$, and $Z {=} w$ (or $Z{=}c$, respectively).
In the provenance graph, 
there are two
\emph{rule nodes} denoting these 
successful derivations of $\rel{Q}(n,s)$ by rule $r_1$. 
A derivation is successful if all goals in the body evaluate to true, i.e., a successful \emph{rule node} is connected to 
successful \emph{goal nodes} 
(e.g., $r_1$ is connected to $g_1^1$, the $1^{st}$ goal in the rule's body). 
A positive (negated) goal is successful 
if the corresponding tuple is (is not) in the database.
Thus, a successful goal node
is connected to the node corresponding to the existing (green) or 
missing (red) tuple justifying the goal, respectively.
\end{Example}


Supporting negation and missing answers is quite challenging, because we need to enumerate all potential ways of deriving a missing answer (or intermediate result corresponding to a negated subgoal) and explain why each of these derivations has failed. An important question in this respect is how to bound the set of missing answers to be considered. 
Under the open world assumption, provenance would be infinite. Using the closed world assumption, only values that exist in the database or are postulated by the query are used to 
construct missing tuples. As is customary in Datalog, we refer to this set of values as the active domain $\adom{I}$ of a database instance $I$. 
We will revisit the assumption that all derivations with constants from $\adom{I}$ are meaningful 
later on.

\begin{Example}
The explanation for $\whynotq \rel{Q}(s,n)$
is shown in  Fig.\,\ref{fig:exam-pg-whynot-chicago-seattle}, 
i.e., why it is not true that New York is reachable from Seattle with exactly one transfer, but not directly.
The tuple $\rel{Q}(s,n)$ is missing from the query result
because all potential ways of deriving this tuple through the rule $r_1$ have failed.
In this example, $\adom{I} {=} \{c,n,s,w\}$ and, thus, there exist four failed derivations of $\rel{Q}(s,n)$ choosing either of these cities as the intermediate stop between Seattle and New York.
A rule derivation fails if at least one goal in the body evaluates to false.
In the provenance graph, only failed goals are connected to the failed rule
derivations explaining missing answers.
Failed positive goals in the body of a failed rule are explained by missing
tuples (red \emph{tuple nodes}). 
For instance, we cannot reach New York from Seattle with an intermediate stop in
Washington DC (the first failed
rule derivation from the left in Fig.\,\ref{fig:exam-pg-whynot-chicago-seattle}) 
because there exists no 
connection from Seattle to Washington DC 
(a \emph{tuple node} $\rel{T}(s,w)$ in red), 
and Washington DC to New York 
(a \emph{tuple node} $\rel{T}(w,n)$ in red). Note that 
the successful goal $\dlNeg \rel{T}(s,n)$ (there is no direct connection from Seattle to New York) does not contribute to the failure of this derivation and, thus, is not part of the explanation.
A failed negated goal is explained by an existing tuple in the database.  
That is, if a tuple $(s,n)$ would exist in the $\rel{Train}$ relation, 
then an additional failed \emph{goal node} $g_1^3 (s,n)$ would be part of the explanation and be connected to each failed rule derivation. 
\end{Example}





\mypartitle{Overview and Contributions}
\label{sec:contribution}
\BG{\emph{Provenance games}~\cite{KL13}, a game-theoretical formalization of provenance for first-order queries 
(i.e., non-recursive Datalog with negation), satisfies these desiderata. 
However, a provenance game for a query $Q$ encodes the provenance of every result and every missing answer of $Q$.}
\BG{Furthermore, the provenance game approach~\cite{KL13} requires the evaluation of a recursive Datalog program with negation under well-founded semantics over a game graph that to construct the provenance game.}
\emph{Provenance games}~\cite{KL13}, a game-theoretical formalization of provenance for first-order (FO) queries,
also supports queries with negation. 
However, 
the approach is computationally expensive, because it requires instantiation of a provenance graph explaining all answers and missing answers.
For instance,
the provenance graph produced by this approach 
for our 
toy example already contains  
more than $64$ (=$4^3$) nodes (i.e., only counting nodes corresponding to rule 
derivations), 
 because there are $4^3$ ways of binding values from $\adom{I} {=} \{c,n,s,w\}$  
to the $3$ variables ($X$, $Y$, and $Z$) of the rule $r_1$.
Typically, most of the nodes 
will not end up being part of the explanation for the user's provenance question.
%
To efficiently compute the explanation, 
we introduce 
a new 
Datalog program 
which computes part of 
the provenance graph 
of an explanation 
bottom-up.
Evaluating this program over instance $I$ 
returns the edge relation of an explanation. 

The main driver of our approach is 
a rewriting of Datalog rules that captures successful and failed rule derivations. 
This rewriting replaces the rules of the program with so-called \textit{firing rules}. 
Firing rules for positive queries were first introduced in~\cite{kohler2012declarative}. These rules are similar to other query instrumentation techniques that have been used for provenance capture such as the rewrite rules of Perm~\cite{GM13}. One of our major contributions is to extend this concept for negation and failed rule derivations which is needed to support Datalog with negation 
and missing answers.
Firing rules provide sufficient information for constructing explanations. However, to make this approach efficient, we need to avoid capturing rule derivations that will not contribute an explanation, 
i.e., they are not connected to 
the nodes corresponding to the provenance question 
in the provenance graph. We achieve this by propagating information from the user's provenance question  
throughout the query 
to prune rule derivations early on 
1) if they do not agree with the constants in the question 
or 2) if we can determine that based on their success/failure status they cannot be part of the explanation. 
For instance, in our 
running example, $\rel{Q}(n,s)$ can only be
connected to successful derivations 
of the rule $r_1$ with $X {=} n$ and $Y {=} s$.
We have presented a proof-of-concept version of our approach as a poster~\cite{LS16}.
%
%
Our main contributions are: 
\begin{itemize}
\item We introduce a provenance graph model for full first-order (FO)
  queries, expressed  as \emph{non-recursive Datalog queries with negation}
  (or \emph{Datalog} for short).
\item We extend the concept of firing rules to support negation and missing answers. 
\item We present an efficient method for computing explanations
      to provenance questions. 
      Unlike the solution in~\cite{KL13}, our approach
      avoids unnecessary work by focusing the computation on 
      relevant parts of the provenance graph.  
 \item We prove the correctness 
       of our algorithm that computes the 
       explanation to a provenance question. 
\item We present a full implementation of our approach in the GProM~\cite{AG14} system.
      Using this system, we compile Datalog into relational algebra expressions, 
      and translate the expressions 
      into SQL code that
      can be executed by a standard relational database backend.
\end{itemize}

The remainder of this paper is organized as follows.
We formally  define the problem in Sec.\,\ref{sec:probl-defin-backgr}, 
 discuss related work in Sec.\,\ref{sec:rel-work}, and
present our approach for computing explanations 
 in Sec.\,\ref{sec:compute-gp}.
We then discuss our implementation (Sec.\,\ref{sec:transl-into-relat}), present 
experiments (Sec.\,\ref{sec:experiments}), and conclude in Sec.\,\ref{sec:concl}.


\section{Problem Definition}
\label{sec:probl-defin-backgr}

We now formally define the problem addressed in this work: how to find the subgraph of a 
provenance graph for a given query (input program) 
$P$ and instance $I$ that explains existence/absence of a 
tuple in/from the result of $P$. 

\subsection{Datalog}
\label{sec:datalog}

A Datalog program $P$ consists of a finite set of rules $r_i: \rel{R}(\vec{X}) \dlImp \rel{R_1}(\vec{X_1}), \ldots,$ $\rel{R_n}(\vec{X_n})$ where $\vec{X_j}$ denotes a tuple of variables and/or constants. 
We assume that the rules of a program 
are labeled $r_1$ to $r_m$. $\rel{R}(\vec{X})$ is the \emph{head} of
the rule, denoted $\headOf{r_i}$, and $\rel{R_1}(\vec{X_1}), \ldots,
\rel{R_n}(\vec{X_n})$ is the \emph{body} (each $\rel{R_j}(\vec{X_j})$
is a \emph{goal}). 
We use $\varsOf{r_i}$ to denote 
the set of variables in $r_i$. 
In this paper, consider non-recursive Datalog with negation, so
 goals $\rel{R_j}(\vec{X_j})$ in the body are \emph{literals}, i.e.,
 atoms $\rel{A}(\vec{X_j})$ or their negation $\neg
 \rel{A}(\vec{X_j})$.  Recursion is not allowed. 
All rules $r$ of a program have to be \emph{safe}, i.e., every
variable in $r$ must occur positively in $r$'s body (thus, head
variables and variables in negated goals must also occur in a positive goal).
For example, Fig.\,\ref{fig:running-example-db} shows a Datalog query with a single rule $r_1$. 
Here, 
$\headOf{r_1}$ is $\rel{Q}(X,Y)$ and 
$\varsOf{r_1}$ is $\{X,Y,Z\}$. The rule is safe since the head
variables and the variables in the negated goal also occur positively
in the body ($X$ and $Y$ in both cases).
The set of relations in the schema 
 over which $P$ is defined is referred to as the extensional database
 (EDB), while relations defined through rules in $P$ form the
 intensional database (IDB), i.e., the IDB relations are those defined in the head of rules.
We require that 
$P$ has a distinguished IDB relation $Q$, called the \emph{answer} relation. Given 
$P$ and instance $I$, we use $P(I)$ to denote the result of $P$ evaluated over $I$. Note that $P(I)$ includes the instance $I$, i.e., all EDB atoms that are true in 
$I$. 
For an EDB or IDB predicate $R$, we use $R(I)$ 
to denote the instance of $R$ computed by $P$ and $R(t) \in P(I)$ to denote that $t \in R(I)$ according to $P$.

We use $\adom{I}$ to denote the active domain of instance $I$, i.e., the set of all constants that occur in $I$. Similarly, we use $\adom{\rel{R.A}}$ to denote the active domain of attribute $A$ of relation $\rel{R}$.
In the following, we make use of the concept of a rule derivation. 
A \emph{derivation} 
of a rule $r$ is an assignment of variables in $r$ to constants from $\adom{I}$.
For a rule with $n$ variables, we use $r(c_1, \ldots, c_n)$ to denote the 
derivation 
that is the result of binding $X_i {=} c_i$. We call a derivation 
\textit{successful} wrt.\ an instance $I$ if each atom 
in the body of the rule is true in $I$ and \textit{failed} otherwise. 

\subsection{Negation and Domains}
\label{sec:domains}

To be able to explain why a tuple is missing, 
we have to enumerate all failed derivations of this tuple and, for each such derivation, explain why it failed. 
As mentioned in Sec.~\ref{sec:intro}, the question is what is a feasible set of potential answers to be considered as missing. 
While the size of why-not provenance is typically infinite under the open world assumption, 
we have to decide how to bound the set of missing answers in the closed world assumption. 
We propose a simple, yet general, solution 
by assuming that each attribute 
of an IDB or 
EDB relation has an associated domain. 

\begin{Definition}[Domain Assignment]
\label{def:definition1}
  Let $S = \{\rel{R_1}, \ldots, \rel{R_n}\}$ be a database schema where each $\rel{R_i(A_{1}, \ldots, A_{m})}$ is a relation schema. Given an instance $I$ of $S$, a \emph{domain assignment} $\domA$ is a function that associates with each attribute $\rel{R.A}$ a domain of values.
We require 
$\domA(\rel{R.A}) \supseteq \adom{\rel{R.A}}$. 
\end{Definition}

In our approach, the user specifies each $\domA(\rel{R.A})$ as a query 
$\domA_{\rel{R.A}}$ that returns the set of admissible values for the domain of attribute $\rel{R.A}$.
We provide reasonable defaults to avoid forcing the user to specify $\domA$ for every attribute, e.g., $\domA(\rel{R.A}) {=} \adom{\rel{R.A}}$ for unspecified $\domA_{\rel{R.A}}$. 
These associated domains fulfill two purposes: 1) to reduce the size of explanations 
and 2) to avoid semantically meaningless answers.
For instance, if there would exist another attribute $\rel{Price}$ in the relation $\rel{Train}$ in Fig.\,\ref{fig:running-example-db}, 
then $\adom{I}$ would also include all the values that appear in this attribute. 
Thus, 
some failed rule derivations 
for $r_1$ would assign prices to the variable representing intermediate stops.
Different attributes may represent the same type of entity (e.g., $\rel{fromCity}$ and $\rel{toCity}$ in our example) and, thus, it would make sense to use their combined 
domain values when constructing missing answers. 
For now, we leave it up to the user to specify attribute domains. Using techniques for discovering semantic relationships among attributes to automatically determine feasible attribute domains is an interesting avenue for future work.

When defining provenance graphs in the following, we are only interested in rule derivations that use constants from the associated domains of attributes accessed by the rule. Given a rule $r$ and variable $X$ used in this rule, let $\attrsOf{r,X}$ denote the set of attributes that variable $X$ is bound to in the body of the rule. For instance, in Fig.\,\ref{fig:running-example-db}, $\attrsOf{r_1,Z} {=} \{\rel{Train.fromCity}, \rel{Train.toCity}\}$. 
We say a rule derivation 
$r(c_1, \ldots, c_n)$ is \emph{domain grounded} iff $c_i \in \bigcap_{A \in \attrsOf{r,X_i}} \domA(A)$ for all $i \in \{1, \ldots, n\}$. 

\subsection{Provenance Graphs}
\label{sec:provenance-graphs}

Provenance graphs 
justify the existence or absence of 
a query result 
based on the success or 
failure of derivations 
using a query's rules, respectively. 
They also explain how the existence or absence of tuples in the database caused derivations to succeed or 
fail, respectively. 
Here, 
we present a constructive definition of provenance graphs that provide this type of justification. Nodes in these graphs  carry two types of labels: 1) a label that determines the node type (tuple, rule, or goal) and additional information, e.g., the arguments and rule identifier of a derivation; 2) the success/failure status of nodes.  

\begin{Definition}[Provenance Graph]\label{def:prov=graph}
  Let $P$ be a first-order (FO) 
query, 
$I$ a database instance, $\domA$ a domain assignment 
for $I$, and $\stringDom$ the domain containing all strings. The \emph{provenance graph} $\provGraph(P,I)$ 
is a graph $(V,E,\nodeLabel,\successLabel)$ with nodes $V$, edges $E$, and node labelling functions $\nodeLabel: V \to \stringDom$ and  $\successLabel: V \to \{\greenT, \redF\}$.
We require that $\forall v,v' \in V: \nodeLabel(v) = \nodeLabel(v') \rightarrow v = v'$.
$\provGraph(P,I)$ is defined as follows: 
\begin{compactitem}
  \item \textbf{Tuple nodes:} For each n-ary EDB or IDB predicate $R$ and tuple $(c_1, \ldots,c_n)$ of constants from the associated domains ($c_i \in \domA(\rel{R.A_i})$), there exists a node $v$ labeled $R(c_1, \ldots,c_n)$. $\successLabel(v) = \greenT$ iff $R(c_1, \ldots,c_n) \in P(I)$  
    and $\successLabel(v) = \redF$ otherwise.
  \item \textbf{Rule nodes:} For every successful domain grounded derivation $r_i(c_1, \ldots, c_n)$, there exists a node $v$ in $V$ labeled $r_i(c_1, \ldots,c_n)$ with $\successLabel(v) = \greenT$. For every failed domain grounded derivation $r_i(c_1, \ldots, c_n)$ where $head(r_i$ $(c_1,\ldots,c_n)) \not\in P(I)$, there exists a node $v$ as above but with $\successLabel(v) = \redF$. In both cases, $v$ is connected to the tuple node $\headOf{r_i(c_1, \ldots, c_n)}$.
  \item \textbf{Goal nodes:} Let $v$ be the node corresponding to a derivation $r_i(c_1, \ldots, c_n)$ with $m$ goals. If $\successLabel(v) = \greenT$, then for all $j \in \{1,\ldots,m\}$, $v$ is connected to a goal node $v_j$ labeled $g_i^j$  with $\successLabel(v_j) = \greenT$. 
If $\successLabel(v) = \redF$, then for all $j \in \{1,\ldots,m\}$, $v$ is connected to a goal node $v_j$ with $\successLabel(v_j) = \redF$
if the $j^{th}$ goal is failed in $r_i(c_1, \ldots, c_n)$. 
Each goal is connected to the corresponding tuple node. 
  \end{compactitem}
\end{Definition}
%
Our provenance graphs model query evaluation by construction. 
A tuple node $R(t)$ is successful in $\provGraph(P,I)$ iff $R(t) \in P(I)$. This is guaranteed, because each tuple built from values of the associated 
domain 
exists as a node $v$ in the graph and its label $\successLabel(v)$ is decided based on $R(t) \in P(I)$.  Furthermore, there exists a successful rule node $r(\vec{c}) \in \provGraph(P,I)$ iff the derivation $r(\vec{c})$ succeeds for $I$. Likewise, 
a failed rule node $r(\vec{c})$ exists iff the derivation $r(\vec{c})$ is failed over $I$ and $head(r(\vec{c})) \not\in P(I)$. 
Fig.\,\ref{fig:exam-pg-why-NY-seattle} and\,\ref{fig:exam-pg-whynot-chicago-seattle} show subgraphs of $\provGraph(P,I)$ for the query 
from Fig.\,\ref{fig:running-example-db}. Since $\rel{Q}(n,s) \in P(I)$ (Fig.\,\ref{fig:exam-pg-why-NY-seattle}), this tuple node is connected to all successful 
 derivations with $\rel{Q}(n,s)$ in the head which in turn are connected to goal nodes for each of the three goals of rule $r_1$. 
$\rel{Q}(s,n) \notin P(I)$ (Fig.\,\ref{fig:exam-pg-whynot-chicago-seattle}) and, thus, its node is connected to all failed derivations 
with $\rel{Q}(s,n)$ as a head. Here, we have 
assumed that all cities can be considered as starting and end points of missing train connections, i.e., both $\domA(\rel{T.fromCity})$ and $\domA(\rel{T.toCity})$ are defined as 
$\adom{\rel{T.fromCity}} \cup \adom{\rel{T.toCity}}$. 
Thus, we have considered derivations $r_1(s,n,Z)$ for $Z \in \{c,n,s,w\}$. 

An important characteristic of our provenance graphs is that each node $v$ in a graph is uniquely identified by its label $\nodeLabel(v)$. 
Thus, common subexpressions are shared leading to more compact provenance graphs. For instance, observe that the node $g_1^3(n,s)$ is shared by two rule nodes in the explanation shown in Fig\,\ref{fig:exam-pg-why-NY-seattle}.

\subsection{Questions and Explanations}
\label{sec:problem-definition}


Recall that the problem we address in this work is how to explain the existence or absence
of (sets of) tuples using provenance graphs.
Such a set of tuples is called a \textit{provenance question} (\textit{PQ}) in this paper. 
The two questions presented in Example\,\ref{ex:example1} use constants only,
but we also support provenance questions with variables, 
e.g., for a question $\rel{Q}(n,X)$ we would return all explanations for existing or missing
tuples where the first attribute is $n$, i.e., why or why-not a city $X$
can be reached from New York with one transfer, but not directly.
We say a tuple $t'$ of constants  \emph{matches} a tuple $t$ of variables and constants written as  $t' \matches t$
if we can unify $t'$ with $t$, i.e., we can equate $t'$ with $t$ by applying a valuation 
that substitutes variables
in $t$ with constants from $t'$.

\begin{Definition}[Provenance Question]\label{def:question}
Let $P$ be a 
query, 
$I$ an instance, and $Q$ an IDB predicate.
A \emph{provenance question} $\provQ$ is an atom $Q(t)$ where $t = (v_1, \ldots, v_n)$ is a
tuple consisting of variables and constants from the associated domain ($\domA(\rel{Q.A})$ for each attribute $\rel{Q.A}$).
$\whyq Q(t)$ and $\whynotq Q(t)$ restrict the question to 
existing and 
missing tuples 
$t' \matches t$, respectively.
\end{Definition}

In Example\,\ref{ex:example1}, we have presented subgraphs of $\provGraph(P,I)$ 
as \textit{explanations} for \provQ s,
implicitly claiming that these subgraphs are sufficient for explaining the $\provQ$. 
Below, we formally define this type of explanation. 

\begin{Definition}[Explanation]\label{def:explanation}
The \emph{explanation} $\explainq(P,Q(t),I)$ for 
$Q(t)$ (PQ)
according to $P$ and $I$, is the subgraph of $\provGraph(P,I)$ containing
only nodes that are connected to at least one node $Q(t'$) where $t' \matches t$. 
For $\whyq Q(t)$, only existing tuples $t'$ 
are considered to match $t$. 
For $\whynotq Q(t)$ only missing tuples are considered to match $t$.
\end{Definition}

Given this definition of explanation, note that 1) all nodes connected to a tuple node matching the $\provQ$ are relevant for computing this tuple and 2) only nodes connected to this node are relevant for the outcome. 
Consider 
$t'$ 
where $t' \matches t$ for a $Q(t)$ ($\provQ$). 
If $Q(t') \in P(I)$, then all successful derivations with head $Q(t')$ justify the existence of $t'$ and 
these are precisely the rule nodes connected to $Q(t')$ in $\provGraph(P,I)$. 
If $Q(t') \not\in P(I)$, then all derivations with head $Q(t')$ have failed and 
are connected to $Q(t)$ in the provenance graph. Each such derivation is connected to all of its failed goals which are responsible for the failure. 
Now, if a rule body references IDB predicates, 
then the same argument can be applied to reason that all rules directly connected to these tuples 
explain why they (do not) exist. 
Thus, by induction, the explanation contains all relevant tuple and rule nodes that explain the $\provQ$. 

\section{Related Work}
\label{sec:rel-work}


Our provenance graphs have strong connections to other provenance models 
for relational queries, most importantly provenance games and the semiring framework, 
and to approaches for explaining missing answers. 

\mypartitle{Provenance Games}
Provenance games
~\cite{KL13}
model the evaluation of a given query (input program) $P$ over an instance $I$ as a 2-player game in a way that
resembles SLD(NF) resolution. 
If the
position (a node in the game graph) corresponding to a tuple $t$ is won 
(the player starting in this position has a winning strategy), then $t \in P(I)$ 
and if the position is lost, then $t \not\in P(I)$. By virtue of supporting negation, provenance games can uniformly answer why and why-not questions for queries with negation. 
%
%
However, provenance games may be hard to comprehend
for non-power users as they require some background in game theory to be understood, e.g., the won/lost status of derivations in a provenance game is contrary to what may be intuitively expected. 
That is, a rule node is lost if the derivation is successful. 
The status of rule nodes in our provenance graphs matches the intuitive expectation (e.g., the rule node is successful if the derivation exists).  
K\"ohler et al.~\cite{KL13} also present an algorithm that computes 
the provenance game for 
$P$ and 
$I$. However, this approach requires instantiation of the full game graph (which 
enumerates all existing and missing tuples) and evaluation of a recursive Datalog$^\neg$ program over this graph using the well-founded semantics~\cite{FK97}.  In constrast, our approach computes explanations that are succinct subgraphs containing only relevant provenance. 
We use bottom-up evaluation instrumented with firing rules to capture provenance. 
Furthermore, we enable the user to restrict provenance for missing answers. 

\mypartitle{Database Provenance}
Several provenance models for 
database queries have been introduced in related work, e.g., see~\cite{cheney2009provenance,grigoris-tj-simgodrec-2012}). 
The semiring annotation framework generalizes these models for positive relational algebra (and, thus, positive non-recursive Datalog). In this model, tuples in a relation are annotated with elements from a commutative semiring K. 
An essential property of the K-relational model is that semiring $\ProvPoly$, the semiring of \emph{provenance polynomials},  
 generalizes all other semirings.
It has been shown in~\cite{KL13} that provenance games generalize $\ProvPoly$ for positive queries and, thus, all other provenance models expressible as semirings. 
Since our graphs are equivalent to provenance games in the sense that there exist lossless transformations between both models (the discussion is beyond the scope of this paper), our graphs also encode $\ProvPoly$.
Provenance graphs which are similar to our graphs 
restricted to positive queries have been used as graph representations of semiring provenance
(e.g., see~\cite{DG15c,DM14c,grigoris-tj-simgodrec-2012}).
Both our graphs and the boolean circuits representation of semiring provenance~\cite{DM14c} explicitly share common subexpressions in the provenance.
However, while these circuits support recursive queries, they do not support negation.
Exploring the relationship of provenance graphs for queries with negation and m-semirings (semirings with support for set difference) is an interesting avenue for future work.
Justifications for 
logic programs~\cite{PS09} are also 
closely related.



\mypartitle{Why-not and Missing Answers}
Approaches for explaining missing answers 
 can be classified based on whether they explain a missing answer 
based on the query
~\cite{BH14a,BH14,CJ09,TC10} (i.e., which operators filter out tuples that would have contributed to the missing answer) 
or 
based on the input data~\cite{HH10,huang2008provenance} (i.e., what tuples need to be inserted into the database to turn the missing answer into an answer). The missing answer problem was first stated for query-based explanations in the seminal paper by Chapman et al.~\cite{CJ09}. Huang et al.~\cite{huang2008provenance} first introduced an instance-based approach. Since then, several techniques have been developed to exclude spurious explanations, to support larger classes of queries~\cite{HH10},  
 and to support distributed Datalog systems in Y!~\cite{WZ14a}.
The approaches for instance-based explanations (with the exception of Y!) have in common that they treat the missing answer problem as a view update problem: the missing answer is a tuple that should be inserted into a view corresponding to the query and this insert has to be translated to an insert into the database instance. An explanation is then one particular solution to this view update problem. 
In contrast to these previous works, our provenance graphs 
explain missing answers 
by enumerating all failed rule derivations that justify why  the answer 
is not in the result.
Thus, they are arguably a better fit for use cases such as debugging queries, where in addition
 to determining which missing inputs justify a missing answer, the user also needs to understand why derivations have failed. Furthermore, we do support queries with negation.
Importantly, solutions for view update missing answer problems 
can be extracted from our provenance graphs. 
Thus, in a sense, provenance graphs with our approach generalize some of the previous approaches 
(for the class of queries supported, 
e.g., we  do not support aggregation yet). 
Interestingly, recent work has shown that it may be possible to generate 
more concise summaries of provenance games~\cite{GK15,RK14} 
which is particularly useful for negation and missing answers 
to deal with the potentially large size of the resulting provenance. 
Similarly, some missing answer approaches~\cite{HH10} use c-tables to compactly represent sets of missing answers.
These approaches are complementary to our work.




\mypartitle{Computing Provenance Declaratively}
The concept of rewriting a Datalog program using firing rules to capture provenance as variable bindings of rule derivations 
was introduced by K\"ohler et al.~\cite{kohler2012declarative} for provenance-based debugging of positive Datalog queries. 
These rules are also similar to relational implementations of provenance polynomials in Perm~\cite{GM13}, LogicBlox~\cite{GA12}, and Orchestra~\cite{GK07a}. 
Zhou et al.~\cite{ZS10} leverage such rules for the distributed ExSPAN system using either full propagation or reference based provenance. 
An extension of firing rules for negation is the main enabler of our approach. 
\section{Computing Explanations} 
\label{sec:compute-gp}

Recall from Sec.~\ref{sec:intro} 
that our approach 
generates a new Datalog program $\GPProg(P,Q(t),I)$ by rewriting a given query (input program) $P$ 
to return the edge relation of explanation 
$\explainq(P,Q(t),I)$ 
for a provenance question $Q(t)$ ($\provQ$). 
In this section, we explain how to generate the program $\GPProg(P,Q(t),I)$ 
using the following steps:

\mypara{1} We unify the input program $P$ with the $\provQ$ 
	   by propagating the constants in $t$ top-down throughout the program 
	   to be able to later prune irrelevant rule derivations.

\mypara{2} Afterwards, we determine for which nodes in the graph we can infer 
	   their success/failure status 
	   based on the $\provQ$. 
           We model this information as annotations on heads and goals of rules 
 and propagate these annotations top-down. 

\mypara{3} Based on the annotated and unified version 
	   created in the previous
           steps, we generate \textit{firing rules} that capture variable bindings for
           successful and failed rule derivations.

\mypara{4} 
	   To be in the result of one of 
           the firing rules obtained in the previous step is a necessary, 
	   but not sufficient, condition for the corresponding $\provGraph(P,I)$ 
           fragment 
	   to be connected to a node matching the $\provQ$. 
           To guarantee that only relevant fragments are returned,
           we introduce additional rules that check connectivity to confirm whether each fragment is connected.

\mypara{5} Finally, we create rules that generate the edge relation of the provenance graph 
	   (i.e., $\explainq(P,Q(t),I)$)
           based on the rule binding information that the firing rules have captured. 

In the following, we will explain each step in detail and illustrate its
application based on the 
question $\whyq \rel{Q}(n,s)$ 
from Example\,\ref{ex:example1}, i.e., 
why is New York connected to Seattle via a train connection with one intermediate stop, but is not directly connected to Seattle. 






\subsection{Unify the Program with PQ}
\label{sec:unify-program-with}

The node $\rel{Q}(n,s)$ 
in the provenance graph (Fig.\,\ref{fig:exam-pg-why-NY-seattle}) 
is only connected to rule derivations which return
$\rel{Q}(n,s)$. 
For instance, if variable $X$ is bound to another city $x$ (e.g., Chicago) in a
derivation of the rule $r_1$, then this rule cannot return the tuple $(n,s)$.  
This reasoning can be applied recursively to replace variables in rules with
constants. That is, we unify the rules in the program top-down with the $\provQ$. 
This step may produce multiple duplicates of a rule with different bindings. 
We use superscripts to make explicit the variable binding used by a replica of a rule.

\begin{Example}
Given the question $\whyq \rel{Q}(n,s)$, 
we unify the single rule
$r_1$ using the assignment $(X{=}n,Y{=}s)$:
\begin{align*}
r_1^{(X=n,Y=s)}:  \rel{Q}(n,s) &\dlImp \rel{T}(n,Z), \rel{T}(Z,s), \dlNeg \rel{T}(n,s)
\end{align*}

\end{Example}

 We may have to create multiple partially
unified versions of a rule or an EDB atom.  
For example, to explore successful derivations of $\rel{Q}(n,s)$ we are interested in both train connections from New York to some city ($\rel{T}(n,Z)$) and from any city to Seattle ($\rel{T}(Z,s)$). Furthermore, we need to know whether there is a direct connection from New York to Seattle ($\rel{T}(n,s)$).
The general idea of this step is motivated by~\cite{BMSU86} which introduced ``magic sets'', an approach for rewriting logical rules to cut down irrelevant facts 
by using additional predicates called ``magic predicates''. Similar techniques exist in standard relational optimization under the name of predicate move around.
We use the technique of propagating variable bindings in the query
to restrict the computation 
based on the user's interest.
This approach is correct because if we bind a variable in the head of rule, then only rule derivations that agree with this binding can derive tuples that agree with this binding.
Based on this unification step, 
we know which 
bindings may produce fragments of $\provGraph(P,I)$ 
that are relevant 
for explaining  the $\provQ$. 
The algorithm implementing this step is given as Algorithm\,\ref{alg:unify-program}.
%
\begin{algorithm}[t]
{\small
  \begin{algorithmic}[1]
    \Procedure{UnifyProgram}{$P$, $Q(t)$}
      \State $todo \gets [Q(t)]$
      \State $done \gets \{\}$
      \State $\unProg{P} = []$
      \While {$todo \neq []$}
        \State $a \gets \Call{pop}{todo}$
        \State $\Call{insert}{done,a}$ 
        \State $rules \gets \Call{getRulesForAtom}{P,a}$
        \ForAll {$r \in rules$}
          \State $unRule \gets \Call{unifyRule}{r,a}$
          \State $\unProg{P} \gets \unProg{P} \listconcat unRule$
          \ForAll {$g \in \bodyOf{unRule}$}
            \If {$g \not\in done$}
              $todo \gets todo \listconcat g$
            \EndIf
          \EndFor
        \EndFor
      \EndWhile
      \State \Return $\unProg{P}$
    \EndProcedure
  \end{algorithmic}
}
  \caption{Unify Program With $\provQ$}\label{alg:unify-program}
\end{algorithm}



\subsection{Add Annotations based on Success/Failure}
\label{sec:part-inst-game}

For $\whyq$ and $\whynotq$ questions, we only consider tuples that are existing and 
missing, respectively. 
Based on this information, we can infer restrictions on the success/failure status of nodes in 
the provenance graph
that are connected to $\provQ$ 
node(s) (belong to the explanation).
We store these restrictions as annotations $\greenT$, $\redF$, and $\redF/\greenT$
on heads and goals of rules.
Here, $\greenT$ indicates that we are only interested in successful nodes, $\redF$ that we are only interested in failed nodes, and $\redF/\greenT$
that we are interested in both. 
These annotations are determined using a top-down propagation seeded with the $\provQ$.

\begin{Example} \label{ex:example3}
Continuing with our running example question $\whyq \rel{Q}(n,s)$, we know that $\rel{Q}(n,s)$ 
  is successful because the tuple 
 is in the result 
(Fig.\,\ref{fig:running-example-db}). 
This implies that only successful rule nodes and their successful goal nodes can be connected to this tuple node. 
Note that this annotation does not imply that the rule $r_1$ 
would be successful for every $Z$ 
(i.e., every intermediate stop between New York and Seattle). 
It only indicates that it is sufficient to focus on successful rule derivations 
since failed ones cannot be connected to $\rel{Q}(n,s)$. 
\begin{align*}
r_1^{(X=n,Y=s), \greenT}: \rel{Q}(n,s)^{\greenT} \dlImp \rel{T}(n,Z)^{\greenT}, \rel{T}(Z,s)^{\greenT}, \dlNeg \rel{T}(n,s)^{\greenT}
\end{align*}
%
We now propagate the annotations of the goals in $r_1$ throughout the program. 
That is, for any goal that is an IDB predicate, we propagate its annotation 
to the head of all rules deriving the goal's predicate 
and, then, propagate these annotations  
to the corresponding rule bodies.  
Note that the inverted 
annotation is propagated for negated goals. For instance, if $\rel{T}$ would be an IDB predicate, then the annotation on the goal
 $\dlNeg \rel{T}(n,s)^{\greenT}$ would be propagated as follows.  We would  annotate the head of all rules deriving $\rel{T}(n,s)$ with $\redF$, because $\rel{Q}(n,s)$ can only exist if $\rel{T}(n,s)$ does not exist.
%
\end{Example}

Partially unified atoms (such as $\rel{T}(n,Z)$) may occur in both negative and positive goals of the rules of the program. 
We denote such atoms  
%
%
 using a $\redF/\greenT$ annotation. 
The use of these annotations will become more clear in the next subsection when
we introduce firing rules.
The pseudocode for the algorithm that determines these annotations 
is given as Algorithm\,\ref{alg:annot-program}. 
In short:
\smallskip
\begin{compactenum}
\item Annotate the head of all rules deriving tuples matching 
the question with $\greenT$ (why) or $\redF$ (why-not).
\item Repeat the following steps until a fixpoint is reached:
  \begin{compactenum}
  \item Propagate the annotation of a rule head to goals in the rule body as follows: propagate $\greenT$ for $\greenT$ annotated heads and $\redF/\greenT$ for $\redF$ annotated heads.
  \item For each annotated goal in the rule body, propagate its annotation to all rules that have this atom in the head. 
For negated goals, unless the annotation is $\redF/\greenT$, we propagate the inverted annotation (e.g., $\redF$ for $\greenT$) to the head of rules deriving the goal's predicate.
  \end{compactenum}
\end{compactenum}
\smallskip
\begin{algorithm}[t]
{\small
  \begin{algorithmic}[1]
    \Procedure{AnnotProgram}{$\unProg{P}$, $Q(t)$}
      \State $state \gets \qType(Q(t))$
      \State $todo \gets [Q(t)^{state}]$
      \State $done \gets \{\}$
      \State $\adProg{P} = []$
      \While {$todo \neq []$}
        \State $a \gets \Call{pop}{todo}$
        \State $state \gets \qType(a)$
        \State $\Call{insert}{done,a}$ 
        \State $rules \gets \Call{getUnRulesForAtom}{P,a}$
        \ForAll {$r \in rules$}
          \State $annotRule \gets \Call{annotRule}{r,state}$
          \State $\adProg{P} \gets \adProg{P} \listconcat annotRule$
          \ForAll {$g \in \bodyOf{annotRule}$}
              \If {state = \redF}
                  \State $newstate \gets \redF/\greenT$
              \Else
                  \State $newstate \gets state$
              \EndIf
              \If {\Call{isNegated}{g}}
              \State $newstate \gets \Call{switchState}{state}$
              \EndIf
            \If {$g^{newstate} \not\in done \wedge \Call{isIDB}{g}$}
              \State $todo \gets todo \listconcat g^{newstate}$
            \EndIf
          \EndFor
        \EndFor
      \EndWhile
      \ForAll {$r \in \adProg{P}$}
         \If {$\qType(r) = \redF/\greenT$}
             \State $\adProg{P} \gets \Call{removeAnnotatedRules}{\adProg{P},r,\{\redF,\greenT\}}$
         \EndIf
      \EndFor
      \State \Return $\adProg{P}$
\EndProcedure
  \end{algorithmic}
}
  \caption{Success/Failure Annotations}\label{alg:annot-program}
\end{algorithm}



\subsection{Creating Firing Rules}
\label{sec:creat-firer-rules}


To be able to compute the relevant subgraph of $\provGraph(P,I)$ 
(the explanation) for the provenance question $\provQ$, 
we need to determine successful and/or failed rule derivations.
Each rule derivation paired with the information whether it is successful
over the given database instance (and which goals are failed in case it is not
successful) corresponds to a certain subgraph.
Successful rule derivations are always part of $\provGraph(P,I)$ 
for a given query (input program) $P$ 
whereas failed
rule derivations only appear 
if the tuple in the head 
failed,
i.e., there are no successful derivations of any rule with
this head. 
To capture the variable bindings of successful/failed rule derivations, we create
``firing rules''.
For successful rule derivations, a firing rule consists of the body of the rule 
(but using the firing version of each predicate in the body) 
and
a new head predicate that contains all variables used in the rule. 
In this way, the firing rule captures all the variable bindings of a rule derivation. 
Furthermore, for each IDB predicate $R$ that occurs as a head of a rule $r$, we create a firing rule that has the firing version of predicate $R$ in the head and a firing version of the rules $r$ deriving the predicate in the body. 
For EDB predicates, we create firing rules that have the firing version of the predicate in the head and the EDB predicate in the body.  
\begin{figure}[t]\centering
 $\,$\\[-3mm] 
  \begin{minipage}{1\linewidth}
    \begin{align*}
      \fire{Q}{}{\greenT}(n,s) &\dlImp \fire{r_1}{Q}{\greenT}(n,s,Z)\\[1mm]
\fire{r_1}{Q}{\greenT}(n,s,Z) &\dlImp \fire{T}{}{\greenT}(n,Z), \fire{T}{}{\greenT}(Z,s), \fire{T}{}{\redF}(n,s)\\[1mm]
\fire{T}{}{\greenT}(n,Z) &\dlImp \rel{T}(n,Z)\\
\fire{T}{}{\greenT}(Z,s) &\dlImp \rel{T}(Z,s)\\
\fire{T}{}{\redF}(n,s) &\dlImp \dlNeg \rel{T}(n,s)
    \end{align*}\\[-10mm]
    \caption{Example firing rules for $\whyq \rel{Q}(n,s)$}
    \label{fig:exam-fire-rules}
  \end{minipage}

\end{figure}

\begin{Example}
\label{ex:example4}
Consider the annotated program in Example\,\ref{ex:example3} for the question $\whyq \rel{Q}(n,s)$. 
We generate the firing rules shown in Fig.\,\ref{fig:exam-fire-rules}. 
The firing rule for $r_1^{(X=n,Y=s), \greenT}$ (the second rule from the top) 
is derived 
from the rule $r_1$ 
by adding $Z$
(the only existential variable) to the head, 
renaming the head predicate as $\fire{r_1}{Q}{\greenT}$, 
and replacing each goal with its firing version 
(e.g., $\fire{T}{}{\greenT}$ for the two positive goals and $\fire{T}{}{\redF}$ for the negated goal). 
Note that negated goals are replaced with firing rules that have inverted annotations
(e.g., the goal $\dlNeg \rel{T}(n,s)^{\greenT}$ is replaced with $\fire{T}{}{\redF}(n,s)$). 
Furthermore, we introduce firing rules for EDB tuples 
(the three rules from the bottom in Fig.\,\ref{fig:exam-fire-rules}) 
\end{Example}

As mentioned in Sec.~\ref{sec:rel-work}, firing rules for successful rule derivations have been used for declarative debugging of positive Datalog programs~\cite{kohler2012declarative} and, for non-recursive queries, are essentially equivalent to rewrite rules that instrument a query to compute provenance polynomials~\cite{AG14,GM13}.
We extend firing rules to support queries with negation and capture missing answers.
To construct a $\provGraph(P,I)$ 
fragment corresponding to a missing tuple, 
we need to find failed rule derivations with the tuple in the head 
and ensure that no successful derivations exist with this head 
(otherwise, we may capture irrelevant failed derivations of existing tuples).
In addition, we need to determine which goals are failed for each failed rule derivation because only failed goals are connected to the node representing the failed rule derivation in the 
provenance graph. 
To capture this information, we add additional boolean variables - $V_i$ for goal $g^i$ - to the head of a firing rule that record for each goal whether it failed or not.
The body of a firing rule for failed rule derivations is created by 
replacing every goal in the body with its $\redF/\greenT$ firing version, and
adding the firing version of the negated head 
to the body (to ensure that only bindings for missing tuples are captured). Firing rules capturing failed derivations use the $\redF/\greenT$ firing versions of their goals because not all goals of a failed derivation have to be failed and 
the failure status determines whether the corresponding goal node is part of the explanation.
A firing rule capturing missing IDB or EDB tuples may not be safe, i.e., it may contain variables that only occur in negated goals. In fact, these variables should be restricted to the associated domains for the attributes 
the variables are bound to. Since the associated domain $\domA$ 
for an attribute $\rel{R.A}$ is given as an unary query $\domA_\rel{R.A}$, we can use these queries directly in firing rules to restrict the values 
the variable is bound to. This is how we ensure that only missing answers 
formed from the associated domains are considered and that firing rules are always safe. 

\begin{figure}[t]\centering
 $\,$\\[-3mm]  
\begin{minipage}{1\linewidth}
    \begin{align*}
\fire{Q}{}{\redF}(s,n) &\dlImp \dlNeg \fire{Q}{}{\greenT}(s,n)\\
\fire{Q}{}{\greenT}(s,n) &\dlImp \fire{r_1}{Q}{\greenT}(s,n,Z)\\[1mm]
\fire{r_1}{Q}{\redF}(s,n,Z,V_1,V_2, \dlNeg V_3) &\dlImp \fire{Q}{}{\redF}(s,n), \fire{T}{}{\redF/\greenT}(s,Z,V_1),\\ 
                                             &\hspace{5mm}\fire{T}{}{\redF/\greenT}(Z,n,V_2),  \fire{T}{}{\redF/\greenT}(s,n,V_3)\\
\fire{r_1}{Q}{\greenT}(s,n,Z) &\dlImp \fire{T}{}{\greenT}(s,Z), \fire{T}{}{\greenT}(Z,n),
					\fire{T}{}{\redF}(s,n)\\[1mm]
\fire{T}{}{\redF/\greenT}(s,Z,\boolT) &\dlImp \fire{T}{}{\greenT}(s,Z)\\
\fire{T}{}{\redF/\greenT}(s,Z,\boolF) &\dlImp \fire{T}{}{\redF}(s,Z)\\
\fire{T}{}{\greenT}(s,Z) &\dlImp \rel{T}(s,Z)\\
\fire{T}{}{\redF}(s,Z) &\dlImp \domA_{\rel{T.toCity}}(Z), \dlNeg \rel{T}(s,Z)
    \end{align*}\\[-10mm]
    \caption{Example firing rules for $\whynotq \rel{Q}(s,n)$}
    \label{fig:exam-fire-negated}
  \end{minipage}

\end{figure}

\begin{Example}
\label{ex:example5}
Reconsider the question $\whynotq \rel{Q}(s,n)$  from Example\,\ref{ex:example1}. 
The firing rules generated for this question are shown in Fig.\,\ref{fig:exam-fire-negated}.  
We exclude the rules for the second goal 
$\rel{T}(Z,n)$ and the negated goal $\dlNeg \rel{T}(s,n)$ which are analogous to the rules 
for the first goal $\rel{T}(s,Z)$. 
Since $\rel{Q}(s,n)$ is failed (because tuple $(s,n)$, i.e., New York cannot be reachable from Seattle with exactly one transfer, is not in the result), 
we are only interested in failed rule derivations of the rule $r_1$ with $X{=}s$ and $Y{=}n$. 
Furthermore, each rule node in the provenance graph corresponding to such a rule derivation 
will only be connected to failed subgoals.
Thus, 
we need to capture which goals are successful or failed 
for each such failed derivation. 
This can be modelled through boolean variables $V_1$, $V_2$, and $V_3$ 
(since there are three goals in the body) 
that are true if the corresponding goal is successful and false otherwise. 
The firing version $\fire{r_1}{Q}{\redF}(s,n,Z,V_1,V_2,\dlNeg V_3)$ of $r_1$  will contain 
all variable bindings for derivations of $r_1$ such that $\rel{Q}(s,n)$ is the head 
(i.e., guaranteed by adding $\fire{Q}{}{\redF}(s,n)$ to the body), 
the rule derivations are failed, and the $i^{th}$ goal is successful or 
failed 
for this binding iff $V_i$ is true or 
false, respectively.
To produce all these bindings, we need rules capturing successful and failed tuple nodes 
for each subgoal of the rule $r_1$.
We denote such rules using a $\redF/\greenT$ annotation and use a boolean variable (true or false) to record 
whether a tuple exists
(e.g., $\fire{T}{}{\redF/\greenT}(s,Z,\boolT) \dlImp \fire{T}{}{\greenT}(s,Z)$ is 
one of these rules).
Negated goals are dealt with by negating this boolean variable, i.e., the goal is successful if the corresponding tuple does not exist.
For instance, $\fire{T}{}{\redF/\greenT}(s,n,false)$ 
represents the fact that tuple $\rel{T}(s,n)$ (a direct train connection from Seattle to New York) 
is missing. This causes the third goal of $r_1$ to succeed for any derivation where $X{=}s$ and $Y{=}n$.
For each partially unified EDB atom annotated with $\redF/\greenT$, we create four rules: one for existing tuples 
(e.g., $\fire{T}{}{\greenT}(s,Z) \dlImp \rel{T}(s,Z)$),
 one for the failure case (e.g., $\fire{T}{}{\redF}(s,Z) \dlImp \domA_{\rel{T.toCity}}(Z), \dlNeg \rel{T}(s,Z)$), and two for the $\redF/\greenT$ firing version. 
For the failure case, 
we use predicate $\domA_{\rel{T.toCity}}$ to only consider missing tuples $(s,Z)$ where $Z$ is a value from the associated domain of this attribute.
\end{Example}
\begin{algorithm}[t]
{\small
  \begin{algorithmic}[1]
    \Procedure{CreateFiringRules}{$\adProg{P}$, $Q(t)$}
    \State $\fireProg{P} \gets []$
    \State $state \gets typeof(Q(t))$
    \State $todo \gets [Q(t)^{state}]$
    \State $done \gets \{\}$
    \While {$todo \neq []$} \Comment{create rules for a predicate}
       \State $R(t)^{\adornment} \gets \Call{pop}{todo}$
       \State $\Call{insert}{done,R(t)^{\adornment}}$
       \If {\Call{isEDB}{$R$}}
       \State \Call{CreateEDBFiringRule}{$\fireProg{P},R(t)^{\adornment}$}
       \Else
       \State \Call{CreateIDBNegRule}{$\fireProg{P},R(t)^{\adornment}$}
       \State $rules \gets \Call{getRules}{R(t)^{\adornment}}$
       \ForAll {$r \in rules$} \Comment{create firing rule for $r$}
         \State $args \gets (\varsOf{r} - \varsOf{\headOf{r}})$
         \State $args \gets \argsOf{\headOf{r}} \listconcat args$
         \State \Call{CreateIDBPosRule}{$\fireProg{P},R(t)^{\adornment},r,args$}
       \State \Call{CreateIDBFiringRule}{$\fireProg{P},R(t)^{\adornment},r,args$}
       \EndFor
       \EndIf
    \EndWhile
    \State \Return $\fireProg{P}$
\EndProcedure
  \end{algorithmic}
}
  \caption{Create Firing Rules}\label{alg:firing-rules}
\end{algorithm}


\begin{algorithm}[h]

{\small

  \begin{algorithmic}[1]
    \Procedure{CreateEDBFiringRule}{$\fireProg{P}$, $R(t)^{\adornment}$}
           \State $[X_1,\ldots,X_n] \gets \varsOf{t}$
              \State $r_{\greenT}  \gets \fire{R}{}{\greenT}(t) \dlImp R(t)$
              \State $r_{\redF}  \gets \fire{R}{}{\redF}(t) \dlImp \domA({X_1}), \ldots, \domA({X_n}), \dlNeg R(t)$
              \State $r_{\redF/\greenT-1}  \gets \fire{R}{}{\redF/\greenT}(t,\boolT) \dlImp \fire{R}{}{\greenT}(t)$
              \State $r_{\redF/\greenT-2}  \gets \fire{R}{}{\redF/\greenT}(t,\boolF) \dlImp \fire{R}{}{\redF}(t)$
           \If {$\adornment = \greenT$}
              \State $\fireProg{P} \gets \fireProg{P} \listconcat r_{\greenT}$      
           \ElsIf {$\adornment = \redF$}
              \State $\fireProg{P} \gets \fireProg{P} \listconcat r_{\greenT} \listconcat r_{\redF}$      
           \Else
              \State $\fireProg{P} \gets \fireProg{P} \listconcat r_{\greenT} \listconcat r_{\redF} \listconcat r_{\redF/\greenT-1} \listconcat r_{\redF/\greenT-2}$      
           \EndIf
    \EndProcedure
  \end{algorithmic}

  \begin{algorithmic}[1]
    \Procedure{CreateIDBNegRule}{$\fireProg{P}$, $R(t)^{\adornment}$}
       \State $[X_1,\ldots,X_n] \gets \varsOf{t}$
       \If {$\adornment \neq \greenT$}   
           \State $r_{new}  \gets \fire{R}{}{\redF}(t) \dlImp \domA({X_1}), \ldots, \domA({X_n}), \dlNeg \fire{R}{}{\greenT}(t)$
           \State $\fireProg{P} \gets \fireProg{P} \listconcat r_{new}$      
       \EndIf
       \If {$\adornment = \redF/\greenT$}
           \State $r_{\greenT}  \gets \fire{R}{}{\redF/\greenT}(t,true) \dlImp \fire{R}{}{\greenT}(t)$
           \State $r_{\redF}  \gets \fire{R}{}{\redF/\greenT}(t,false) \dlImp \fire{R}{}{\redF}(t)$
           \State $\fireProg{P} \gets \fireProg{P} \listconcat r_{\greenT} \listconcat r_{\redF}$
       \EndIf
    \EndProcedure
  \end{algorithmic}

  \begin{algorithmic}[1]
    \Procedure{CreateIDBPosRule}{$\fireProg{P}$, $R(t)^{\adornment}$, $r$, $args$}
            \State $r_{pred}  \gets \fire{R}{}{\greenT}(t) \dlImp \fire{r}{R}{\greenT}(args)$
          \State $\fireProg{P} \gets \fireProg{P} \listconcat r_{pred}$      
    \EndProcedure
  \end{algorithmic}

  \begin{algorithmic}[1]
    \Procedure{CreateIDBFiringRule}{$\fireProg{P}$, $R(t)^{\adornment}$, $r$, $args$}
         \State $body_{new} \gets []$
         \ForAll {$g_i(\vec{X}) \in \bodyOf{r}$}
            \State $\adornment_{goal} \gets \greenT$
            \If {\Call{isNegated}{g}}
               \State $\adornment_{goal} \gets \redF$
            \EndIf
              \State $g_{new} \gets \fire{\predOf{g_i}}{}{\adornment_{goal}}(\vec{X})$
            \State $body_{new} \gets body_{new} \listconcat g_{new}$
            \If { $ g(\vec{X})^{\greenT} \not\in (done \cup todo) \wedge \adornment = \greenT$ }
               \State $todo \gets todo :: g(\vec{X})^{\adornment_{goal}}$
            \EndIf
         \EndFor
       \State $r_{new} \gets \fire{r}{R}{\greenT}(args) \dlImp body_{new}$
       \State $\fireProg{P} \gets \fireProg{P} \listconcat r_{new}$ 
         \If {$\adornment \neq \greenT$}
             \ForAll {$g_i \in \bodyOf{r}$}
             \If {\Call{isNegated}{$g_i$}}
             \State $args \gets args \listconcat \neg b_i$
             \Else
             \State $args \gets args \listconcat b_i$
             \EndIf
             \EndFor
         \State $body_{new} \gets []$
         \ForAll {$g_i(\vec{X}) \in \bodyOf{r}$}
             \State $g_{new} \gets \fire{\predOf{g_i}}{}{\redF/\greenT}(\vec{X}, b_i)$
            \State $body_{new} \gets body_{new} \listconcat g_{new}$
            \If { $ g(\vec{X})^{\redF/\greenT} \not\in (done \cup todo) $ }
               \State $todo \gets todo :: g(\vec{X})^{\adornment_{goal}}$
            \EndIf
         \EndFor
       \State $r_{new} \gets \fire{r}{R}{\adornment}(args) \dlImp body_{new}$
       \State $\fireProg{P} \gets \fireProg{P} \listconcat r_{new}$ 
         \EndIf
    \EndProcedure
  \end{algorithmic}
}
\caption{Create Firing Rules Subprocedures}\label{alg:firing-rules-sub}
\end{algorithm}
The algorithm that creates the firing rules for an annotated input query is shown as Algorithm\,\ref{alg:firing-rules} (the pseudocode for the subprocedures is given as Algorithm\,\ref{alg:firing-rules-sub}). 
It maintains a list of annotated atoms that need to be processed which is initialized with the $Q(t)$ ($\provQ$). 
For each such atom 
$R(t)^{\adornment}$ (here $\adornment$ is the annotation of the atom), 
it creates firing rules for each rule $r$ that has this atom 
as a head and a positive firing rule for $R(t)$. Furthermore, if the atom is annotated with $\redF/\greenT$ or $\redF$, then additional firing rules are added to capture missing tuples and failed rule derivations. 

\mypartitle{EDB atoms} 
For an EDB atom $R(t)^{\greenT}$, we use procedure \textsc{createEDBFiringRule} to create one rule $\fire{R}{}{\greenT}(t) \dlImp R(t)$ that returns tuples from relation $R$ that matches $t$. For missing tuples ($R(t)^{\redF}$), we extract all variables from $t$ (some arguments may be constants propagated during unification) and create a rule that returns all tuples that can be formed from values of the associated domains of the attributes these variables are bound to and do not exist in $R$. This is achieved by adding goals $\domA{(X_i)}$ 
as explained in Example\,\ref{ex:example5}.

\mypartitle{Rules} 
Consider a rule $r: R(t) \dlImp g_1(\vec{X_1}), \ldots, g_n(\vec{X_n})$.
If the head of $r$ 
is annotated with $\greenT$, then we create a rule with head $\fire{r}{R}{\greenT}(\vec{X})$ where $\vec{X} = \varsOf{r}$ 
and the same body as $r$ except that each goal is replaced with its firing version with appropriate annotation (e.g., $\greenT$ for positive goals). 
For rules annotated with $\redF$ or $\redF/\greenT$, 
we create one additional rule with head $\fire{r}{R}{\redF}(\vec{X},\vec{V})$ where $\vec{X}$ is defined as above, and $\vec{V}$ contains $V_i$ if the $i^{th}$ goal of $r$ is positive and $\dlNeg V_i$ otherwise. The body of this rule contains the $\redF/\greenT$ version of every goal in $r$'s body plus an additional goal $\fire{R}{}{\redF}$ to ensure that the head atom is failed. As an example for this type of rule, consider the third rule from the top in Fig.\,\ref{fig:exam-fire-negated}.
%
%


\mypartitle{IDB atoms} 
For each rule $r$ with head $\rel{R}(t)$, we create a rule $\fire{R}{}{\greenT}(t) \dlImp \fire{r}{R}{\greenT}(\vec{X})$ where $\vec{X}$ is the concatenation of $t$ with all existential variables from the body of $r$.
IDB atoms with $\redF$ or $\redF/\greenT$ annotations are handled in the same way as EDB atoms with these annotations. 
For each $R(t)^{\redF}$, we create a rule with $\dlNeg \fire{R}{}{\greenT}(t)$ in the body using the associated domain queries to restrict variable bindings. For $R(t)^{\redF/\greenT}$, we add two additional rules as shown in Fig.\,\ref{fig:exam-fire-negated} for EDB atoms.

\begin{Theorem}[Correctness of Firing Rules]\label{theo:alg-firingcorrect}
 Let $P$ be an input program, $r$ denote a rule of $P$ with $m$ goals, and $\fireProg{P}$ be the firing version of $P$. We use $r(t) \isSucc P(I)$ to denote that the rule derivation $r(t)$ is successful in the evaluation of program $P$ over $I$. The firing rules for $P$ correctly determine existence of tuples, successful rule derivations, and failed rule derivations for missing answers: 
  \begin{itemize}
  \item $\fire{R}{}{\greenT}(t) \in \fireProg{P}(I) \leftrightarrow R(t) \in P(I)$
  \item $\fire{R}{}{\redF}(t) \in \fireProg{P}(I) \leftrightarrow R(t) \not\in P(I)$
  \item $\fire{r}{}{\greenT}(t) \in \fireProg{P}(I) \leftrightarrow r(t) \isSucc P(I)$
  \item $\fire{r}{}{\redF}(t,\vec{V}) \in \fireProg{P}(I) \leftrightarrow r(t) \isFailed P(I) \wedge \headOf{r(t)} \not\in P(I)$ and for $i \in \{1,\ldots,m\}$ we have that $V_i$ is false iff $i^{th}$ goal fails in $r(t)$.
  \end{itemize}
\end{Theorem}
\begin{proof}
  We prove Theorem~\ref{theo:alg-firingcorrect} by induction over the structure of a program. For the base case, we consider programs of ``depth'' $1$, i.e., only EDB predicates are used in rule bodies. 
Then, we prove correctness for programs of depth $n+1$ based on the correctness of programs of depth $n$.
  We define the depth $\aDepth$ of predicates, rules, and programs as follows: 1) for all EDB predicates $R$, we define $\depthP{R} = 0$; 2) for an IDB predicate $R$, we define $\depthP{R} = \max_{\headOf{r} = R}\depthP{r}$, i.e., the maximal depth among all rules $r$ with $\headOf{r} = R$; 3) the depth of a rule $r$ is $\depthP{r} = \max_{R \in \bodyOf{r}} \depthP{R} + 1$, i.e., the maximal depth of all predicates in its body plus one; 4) the depth of a program $P$ is the maximum depth of its rules: $\depthP{P} = \max_{r \in P} \depthP{r}$.

\mypara{1) Base Case}
Assume that we have a program $P$ with depth $1$, 
e.g., $r: \rel{Q}(\vec{X}) \dlImp \rel{R}(\vec{X_1}), \ldots, \rel{R}(\vec{X_n})$. 
We first prove that the positive and negative versions of firing rules for EDB atoms are correct, 
because only these rules are used for the rules of depth $1$ programs. 
A positive version of EDB firing rule $\fire{R}{}{\greenT}$ 
creates a copy of the input relation \rel{R} and, thus, a tuple $t \in R$ iff $t \in \fire{R}{}{\greenT}$. 
For the negative version $\fire{R}{}{\redF}$, all variables are bound to associated domains $\domA$ 
and it is explicitly checked that $\dlNeg R(\vec{X})$ is true. 
Finally, $\fire{R}{}{\redF/\greenT}$ uses $\fire{R}{}{\greenT}$ and $\fire{R}{}{\redF}$ to determine whether the tuple exists in \rel{R}. 
Since these rules are correct, it follows that $\fire{R}{}{\redF/\greenT}$ is correct. 
The positive firing rule for the rule $r$ (\fire{r}{}{\greenT}) is correct since its body only contains positive and negative EDB firing rules ($\fire{R}{}{\greenT}$ respective $\fire{R}{}{\redF}$) which are already known to be correct.
The correctness of the positive firing version of  a rule's head predicate ($\fire{Q}{}{\greenT}$) follows naturally from the correctness of $\fire{r}{}{\greenT}$.
The negative version of the rule $\fire{r}{}{\redF}(\vec{X},\vec{V})$ contains an additional goal (i.e., $\dlNeg \rel{Q}(\vec{X})$) 
and uses the firing version $\fire{R}{}{\redF/\greenT}$ to return only bindings for failed derivations. 
Since $\fire{R}{}{\redF/\greenT}$ has been proven to be correct, we only need to prove that the negative firing version of the head predicate of $r$ is correct. 
For a head predicate 
with annotation $\redF$, we create two firing rules ($\fire{Q}{r}{\greenT}$ and $\fire{Q}{r}{\redF}$).
The rule $\fire{Q}{}{\greenT}$ was already proven to be correct.
$\fire{Q}{}{\redF}$ is also correct, because it contains only $\fire{Q}{}{\greenT}$ and domain queries in the body which were already proven to be correct. 

\mypara{2) Inductive Step}
It remains to be shown 
that firing rules for programs of depth $n+1$ are correct. 
Assume that firing rules for programs of depth up to $n$ are correct. 
Let $r$ be a firing rule of depth $n+1$ in a program of depth $n+1$.
It follows that $\max_{R \in \bodyOf{r}} \depthP{R} \leq n$,
otherwise $r$ would be of a depth larger than $n+1$. Based on the induction hypothesis, it is guaranteed that the firing rules for all these predicates are correct. 
Using the same argument as in the base case, it follows that the firing rule for $r$ is correct.
\end{proof}

\begin{figure}[t]
  $\,$\\[-7mm]
  \centering
  \begin{align*}
%
\fire{Q}{}{\greenT}(n,s) &\dlImp \fire{r_1}{Q}{\greenT}(n,s,Z)\\
\fire{r_1}{Q}{\greenT}(n,s,Z) &\dlImp \fire{T}{}{\greenT}(n,Z), \fire{T}{}{\greenT}(Z,s), \fire{T}{}{\redF}(n,s)\\
\fireC{r_2}{}{\greenT}{r_1^1}(n,Z) &\dlImp \rel{T}(n,Z), \fire{r_1}{Q}{\greenT}(n,s,Z)\\
\fireC{r_2}{}{\greenT}{r_1^2}(Z,s) &\dlImp \rel{T}(Z,s), \fire{r_1}{Q}{\greenT}(n,s,Z)\\
\fire{T}{}{\redF}(n,s) &\dlImp \dlNeg \rel{T}(n,s)
  \end{align*}\\[-4mm]
  \caption{Example firing rules with connectivity checks}
  \label{fig:exam-connect-joins}
\end{figure}
\subsection{Connectivity Joins}
\label{sec:connectivity-joins}

To be in the result of firing rules is a necessary,
but not sufficient, condition for the corresponding rule node to be
connected to a $\provQ$ node 
in the explanation.
Thus, to guarantee that only nodes connected to the $\provQ$ 
node(s) are returned, 
we have to check whether they are actually connected. 

\begin{Example} \label{ex:example6}
Consider the firing rules for 
$\whyq \rel{Q}(n,s)$ shown in Fig.\,\ref{fig:exam-fire-rules}.
The corresponding rules with connectivity checks are shown in Fig.\,\ref{fig:exam-connect-joins}.
All the rule nodes corresponding to 
$\fire{r_1}{Q}{\greenT}(n,s,Z)$ are guaranteed to be connected to the
PQ node $\rel{Q}(n,s)$. 
For sake of the example, assume that instead of using $\rel{T}$, rule $r_1$ uses  an IDB relation $R$ which is computed using another rule $r_2: \rel{R}(X,Y) \dlImp \rel{T}(X,Y)$.
Consider the firing rule $\fire{r_2}{T}{\greenT}(n,Z) \dlImp \rel{T}(n,Z)$ created based on the second goal of $r_1$.
Some provenance graph fragments computed by this rule may not be connected to $\rel{Q}(n,s)$.
A tuple node $\rel{R}(n,c)$ for a constant $c$ is only connected to the 
node $\rel{Q}(n,s)$ iff it is
part of a successful binding of $r_1$.
That is, for the node $\rel{R}(n,c)$ to be connected, there has to exist another tuple $(c,s)$ in $\rel{R}$.  
We check connectivity to $\rel{Q}(n,s)$ 
one hop at a time.
This is achieved by adding the head of the firing rule for $r_1$ to the body of the firing rule for $r_2$ as shown in Fig.\,\ref{fig:exam-connect-joins} 
(the second and third rule from the bottom). 
We use $\fireC{r_2}{}{\greenT}{r_1^k}(\vec{X})$ to denote the firing rule for $r_2$
connected to the $k^{th}$ goal of rule $r_1$.
Note that, this connectivity check is
unnecessary for rules with only constants (the last rule in Fig.\,\ref{fig:exam-connect-joins}).
%
\end{Example}

\SL{The algo below is potential removal}
\begin{algorithm}[t]
{\small
  \begin{algorithmic}[1]
    \Procedure{AddConnectivityRules}{$\fireProg{P}$, $Q(t)$}
    \State $\fireCProg{P} \gets []$
    \State $paths \gets \Call{pathStartingIn}{\fireProg{P}, Q(t)}$
    \ForAll {$p \in paths$}
       \State $p \gets \Call{filterRuleNodes}{p}$
       \ForAll {$e =(r_i(\vec{X_1})^{\adornment_1},r_j(\vec{X_2})^{\adornment_2}) \in p$}
           \State $goals \gets \Call{getMatchingGoals}{e}$
           \ForAll {$g_k \in goals$}
             \State $g_{new} \gets \Call{unifyHead}{\fire{r_i}{R_1}{\adornment_1}(t_1), g_k, \fire{r_j}{R_2}{\adornment_2}(t_2)}$
             \State $r_{new} \gets \fireC{r_j}{}{\adornment_2}{r_i^k}(t_2) \dlImp \bodyOf{\fire{r_j}{R_2}{\adornment_2}(t_2)}, g_{new}$
             \State $\fireCProg{P} \gets \fireCProg{P} \listconcat r_{new}$
           \EndFor
       \EndFor
    \EndFor
    \State \Return $\fireCProg{P}$
\EndProcedure
  \end{algorithmic}
}
  \caption{Add Connectivity Joins}\label{alg:connectivity-joins}
\end{algorithm}

Algorithm~\ref{alg:connectivity-joins}
 traverses the query's rules starting from the $\provQ$ to find all combinations of rules $r_i$ and $r_j$ such that the head of $r_j$ can be unified with a goal in the body of $r_i$.
We use the subprocedure $\textsc{filterRuleNodes}$ to prune rules containing only constants. 
 For each such pair $(r_i,r_j)$ where the head of $r_j$ corresponds to the $k^{th}$ goal in the body of $r_i$, we create a rule 
 $\fireC{r_j}{}{\greenT}{r_i^k}(\vec{X})$ as follows. We unify the variables of 
the $k^{th}$
goal in the firing rule for $r_i$ with the head variables of the firing rule for $r_j$. All remaining variables of $r_i$ are renamed to avoid name clashes. We then add the unified head of $r_i$ to the body of $r_j$. 
Effectively, these rules check one hop at a time whether rule nodes in the provenance graph are connected to 
the nodes matching the $\provQ$. 

\subsection{Computing the Edge Relation}
\label{sec:comp-edge-relat}

The program created so far captures all the information needed to generate the edge
relation of the graph for the $\provQ$. 
To compute the edge relation, we use Skolem functions to create node 
identifiers.
The identifier of a node captures the type of the node (tuple, 
rule, or 
goal), assignments from variables to constants, and the success/failure status of the
node, 
e.g., a tuple node $\rel{T}(n,s)$ that is successful would be represented as
$\nodeSk{rel}{T}{\greenT}(n,s)$. 
Each rule firing corresponds to a fragment of $\provGraph(P,I)$. 
For example, one such fragment is 
shown 
in Fig.\,\ref{fig:examp-graph-edge} (left).  
Such a substructure is created through a set of rules:
\begin{itemize}
\item One rule creating edges between tuple nodes for the head predicate and 
      rule nodes
\item One rule for each goal connecting a rule node to 
      that goal node (for
      failed rules only the failed goals are connected)
\item One rule creating edges between each goal node and the corresponding EDB tuple node
\end{itemize}
The pseudocode for creating rules that return the edge relation is provided as Algorithm\,\ref{alg:create-move-relations}. 




\begin{figure}[t]
  \begin{minipage}{1.1\linewidth} 
    $\,$\\[-5mm]
    \begin{minipage}{0.4\linewidth}
\resizebox{0.87\textwidth}{!}{\begin{tikzpicture}[>=latex',line join=bevel,line width=0.3mm]
  \definecolor{fillcolor}{rgb}{0.8,1.0,0.8};
  \node (REL_Q_WON_a_c_) at (150bp,284bp) [draw=black,fill=fillcolor,ellipse] {$Q(n,s)$};
  \definecolor{fillcolor}{rgb}{0.8,1.0,0.8};
  \node (RULE_0_LOST_a_c_b_) at (150bp,253bp) [draw=black,fill=fillcolor,rectangle] {$r_1(n,s,Z)$};

  \definecolor{fillcolor}{rgb}{0.8,1.0,0.8};
  \node (GOAL_0_0_WON_a_b_) at (100bp,223bp) [draw=black,fill=fillcolor,rounded corners=.15cm,inner sep=3pt] {$g_{1}^{1}(n,Z)$};
  \definecolor{fillcolor}{rgb}{0.8,1.0,0.8};
  \node (EDB_T_LOST_a_b_) at (93bp,193bp) [draw=black,fill=fillcolor,ellipse] {$T(n,Z)$};

  \definecolor{fillcolor}{rgb}{0.8,1.0,0.8};
  \node (GOAL_0_1_WON_b_c_) at (150bp,223bp) [draw=black,fill=fillcolor,rounded corners=.15cm,inner sep=3pt] {$g_{1}^{2}(Z,s)$};
  \definecolor{fillcolor}{rgb}{0.8,1.0,0.8};
  \node (EDB_T_LOST_b_c_) at (150bp,193bp) [draw=black,fill=fillcolor,ellipse] {$T(Z,s)$};

  \definecolor{fillcolor}{rgb}{0.8,1.0,0.8};
  \node (GOAL_0_2_WON_a_c_) at (200bp,223bp) [draw=black,fill=fillcolor,rounded corners=.15cm,inner sep=3pt] {$g_{1}^{3}(n,s)$};
  \definecolor{fillcolor}{rgb}{1.0,0.51,0.51};
  \node (EDB_T_LOST_a_c_) at (205bp,193bp) [draw=black,fill=fillcolor,ellipse] {$T(n,s)$};

  \draw [->] (RULE_0_LOST_a_c_b_) -> (GOAL_0_2_WON_a_c_);
  \draw [->] (GOAL_0_0_WON_a_b_) -> (EDB_T_LOST_a_b_);

  \draw [->] (RULE_0_LOST_a_c_b_) -> (GOAL_0_0_WON_a_b_);
 
  \draw [->] (GOAL_0_2_WON_a_c_) -> (EDB_T_LOST_a_c_);

  \draw[->] (GOAL_0_1_WON_b_c_) -> (EDB_T_LOST_b_c_);

  \draw [->] (REL_Q_WON_a_c_) -> (RULE_0_LOST_a_c_b_);
  \draw [->] (RULE_0_LOST_a_c_b_) -> (GOAL_0_1_WON_b_c_);
%
\end{tikzpicture}
}
  \end{minipage} \hspace{-10mm}
  \begin{minipage}{0.38\linewidth}
   \small \vspace{3mm} 
    \begin{align*}
      \rel{edge}(\nodeSk{rel}{Q}{\greenT}(n,s), \nodeSk{rule}{r_1}{\greenT}(n,s,Z)) &\dlImp \fire{r_1}{only2hop}{\greenT}(n,s,Z)\\ 
      \rel{edge}(\nodeSk{rule}{r_1}{\greenT}(n,s,Z), \nodeSk{rel}{g_1^1}{\greenT}(n,Z)) &\dlImp \fire{r_1}{only2hop}{\greenT}(n,s,Z)\\ 
      \rel{edge}(\nodeSk{rel}{g_1^1}{\greenT}(n,Z), \nodeSk{rel}{T}{\greenT}(n,Z)) &\dlImp \fire{r_1}{only2hop}{\greenT}(n,s,Z)\\ 
      \rel{edge}(\nodeSk{rel}{g_1^3}{\greenT}(n,s), \nodeSk{rel}{T}{\redF}(n,s)) &\dlImp \fire{r_1}{only2hop}{\greenT}(n,s,Z)\\ 
    \end{align*}
  \end{minipage}
 \end{minipage}
 $\,$\\[-8mm]
 \caption{Example structure/rules for edge relation of explanation}
 \label{fig:examp-graph-edge}
\end{figure}
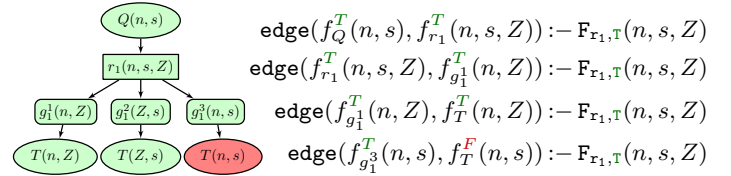

\begin{Example}
  Consider the firing rules with connectivity joins from Example\,\ref{ex:example6}. 
  Some of the rules for creating the edge relation for the explanation sought by
  the user are shown in Fig.\,\ref{fig:examp-graph-edge} (on the right side).
  For example, each edge connecting the tuple node $\rel{Q}(n,s)$ 
  to a successful rule node $r_1(n,s,Z)$ is created by the
  top-most rule, 
  the second rule creates an edge 
  between 
  $r_1(n,s,Z)$ and 
  $g_1^1(n,Z)$, and so on.
\end{Example}

\SLDel{ 
As mentioned in Sec.~\ref{sec:contribution}, by choosing a different set of rules to derive the edge relation, we can derive different types of provenance graphs including provenance games and graphs corresponding to provenance polynomials for positive queries. 
%
For instance, consider one particular type of edges that occurs only in provenance games by not our provenance graphs. A provenance game contain two types of relation nodes: positive and negated. In a provenance game, positive goal nodes are connected to negated tuple nodes which are in turn connected to positive tuple nodes. Furthermore, goal nodes in provenance games record the arguments of the goal. In our example, we could compute edges from  goal $g_1^1$ to a negated tuple node $\neg T$ using a rule $\rel{edge}(\nodeSk{rel}{g_1^1}{\greenT}(n,Z), \nodeSk{rel}{\neg T}{\redF}(n,Z)) \dlImp \fire{r_1}{only2hop}{\greenT}(n,s,Z)$. Note that additional rules are necessary to compute provenance games. For a full set of rules see our technical report~\cite{SY15}.
}

\begin{algorithm}[H]
{\small
  \begin{algorithmic}[1]
    \Procedure{CreateEdgeRelation}{$\fireCProg{P}$, $Q(t)$}
      \State $\moveProg{P} \gets []$
      \State $todo \gets [Q(t)]$
      \State $done \gets \{\}$
      \While {$todo \neq []$}
        \State $R(t)^{\adornment} \gets \Call{pop}{todo}$
        \If {$R(t)^{\adornment} \in done$}
          \State continue
        \EndIf
        \State $done \gets \Call{insert}{done, R(t)^{\adornment}}$
        \State $rules \gets \Call{getRules}{R(t)^{\adornment}}$
        \ForAll {$r \in rules$}
          \If {$\Call{isEDB}{R}$}
            \If {$\adornment = \greenT$}
              \If {$\Call{isNegated}{g}$}
              \State $r_{g \to R} \gets \rel{edge}(\nodeSk{rel}{g}{\greenT}(t'), \nodeSk{rel}{R}{\redF}(t')) \dlImp \fire{r}{R}{\greenT}(args)$
              \Else
              \State $r_{g \to R} \gets \rel{edge}(\nodeSk{rel}{g}{\greenT}(t'), \nodeSk{rel}{R}{\greenT}(t')) \dlImp \fire{r}{R}{\greenT}(args)$
              \EndIf
           \Else
              \State $args_{b} \gets b_1, \ldots, b_{i-1}, false, b_{i+1}, \ldots, b_n$
              \If {$\Call{isNegated}{g}$}
              \State $r_{g \to R} \gets \rel{edge}(\nodeSk{rel}{g}{\redF}(t'), \nodeSk{rel}{R}{\greenT}(t')) \dlImp \fire{r}{R}{\redF}(args, args_{b})$
              \Else
              \State $r_{g \to R} \gets \rel{edge}(\nodeSk{rel}{g}{\redF}(t'), \nodeSk{rel}{R}{\redF}(t')) \dlImp \fire{r}{R}{\redF}(args, args_{b})$
              \EndIf
	   \EndIf
           \State $\moveProg{P} \gets \moveProg{P} \listconcat r_{g \to R}$
          \Else
            \State $\adornment_{r} = \Call{switchState}{\adornment}$
            \State $r_{new} \gets \rel{edge}(\nodeSk{rel}{\predOf{r}}{\adornment}(t), \nodeSk{rel}{r}{\adornment_{r}}(t,\ldots)) \dlImp \fire{r}{R}{\adornment}(t,\ldots)$
            \State $\moveProg{P} \gets \moveProg{P} \listconcat r_{new}$
            \ForAll {$g(t') \in \bodyOf{r}$}
              \If {$\Call{isNegated}{g}$}
              \State $\adornment' \gets \Call{switchState}{\adornment}$
              \Else
              \State $\adornment' \gets \adornment$
              \EndIf
              \State $todo \gets todo \listconcat g(t')^{\adornment'}$
              \If {$\adornment = \greenT$}
              \State $r_{r \to g} \gets \rel{edge}(\nodeSk{rel}{r}{\greenT}(args), \nodeSk{rel}{g}{\greenT}(t')) \dlImp \fire{r}{R}{\greenT}(args)$
              \Else
              \State $args_{b} \gets b_1, \ldots, b_{i-1}, false, b_{i+1}, \ldots, b_n$
              \State $r_{r \to g} \gets \rel{edge}(\nodeSk{rel}{r}{\redF}(args), \nodeSk{rel}{g}{\redF}(t')) \dlImp \fire{r}{R}{\redF}(args, args_{b})$
              \EndIf
              \State $\moveProg{P} \gets \moveProg{P} \listconcat r_{r \to g}$ 
            \EndFor
          \EndIf
        \EndFor
      \EndWhile
      \State \Return $\moveProg{P}$
    \EndProcedure
  \end{algorithmic}
}
  \caption{Create Edge Relation}\label{alg:create-move-relations}
\end{algorithm}

\subsection{Correctness} 
\label{sec:corr-proof-compl}

 We now state correctness of our approach for computing an explanation $\explainq(P,Q(t),I)$ for a provenance question. 

\begin{Theorem}[Correctness]\label{theo:alg-correctness}
The result of a Datalog program $\GPProg(P,Q(t),I)$ that rewrites a query (input program) $P$ over instance $I$ is the edge relation of $\explainq(P,Q(t),I)$. 
\end{Theorem}
\begin{proof}
For Theorem\,\ref{theo:alg-correctness}, 
we prove that 1) only edges from $\provGraph(P,I)$ are returned by the program $\GPProg(P,Q(t),I)$ and 2) the program returns precisely the set of edges of 
 explanation $\explainq(P,Q(t),I)$. 

\mypara{1}
The constant values used as variable binding by the rules creating edges in $\GPProg(P,Q(t),I)$ are either constants that occur in the $\provQ$ or the result of rules which are evaluated over the instance $I$. 
Since only the rules for creating the edge relation create new values (through Skolem functions), it follows that any constant used in constructing a node argument exists in the associated domain. 
Recall that the $\provGraph(P,I)$ 
only  contains nodes with arguments from 
the associated domain. 
Any edge returned by $\GPProg(P,Q(t),I)$ 
is strictly based on the structure of the input program and connects nodes that agree on variable bindings. 
Thus, 
each edge produced by $\GPProg(P,Q(t),I)$ 
will be contained in $\provGraph(P,I)$. 

\mypara{2}
We now prove that the program $\GPProg(P,Q(t),I)$ returns precisely the set of edges of $\explainq(P,Q(t),I)$.
Assume that the $Q(t)$ ($\provQ$) only uses constants (the extension to $\provQ$ which contains variables is immediate).
Consider a rule of an input program of depth $1$ (i.e., only EDB predicates in the body of rules). 
For such a rule node to be connected to the node $Q(t)$, its head variables have to be bound to $t$ (guaranteed by the unification step in Sec.~\ref{sec:unify-program-with}). Since the firing rules are known to be correct, this guarantees that exactly the rule nodes connected to the $\provQ$ node are generated. The propagation of this unification to the firing rules for EDB predicates is correct, because only EDB nodes agreeing with this binding can be connected to such a rule node. However, propagating constants is not sufficient since the firing rule for an EDB predicate (e.g., $R$) may return irrelevant tuples, i.e., tuples that are not part of any rule derivations for $Q(t)$ (e.g., there may not exist EDB tuples 
for other goals in the rule which share variables with the particular goal using predicate $R$). 
This is checked by 
the connectivity joins (Sec.~\ref{sec:connectivity-joins}). If a tuple is returned by a connected firing rule, then the corresponding node is guaranteed to be connected to at least one rule node deriving $\provQ$.
Note that this argument 
does not rely on the fact that predicates in the body of a rule are EDB predicates.  
Thus, we can apply this argument in a proof by induction to show that, given that rules of depth up to $n$ only produce connected rule derivations, the same holds for rules of depth $n+1$.
\end{proof}

\section{Implementation}
\label{sec:transl-into-relat}

\begin{figure}[t]
  $\,$\\[-3mm]
  \centering
  \includegraphics[width=0.78\linewidth]{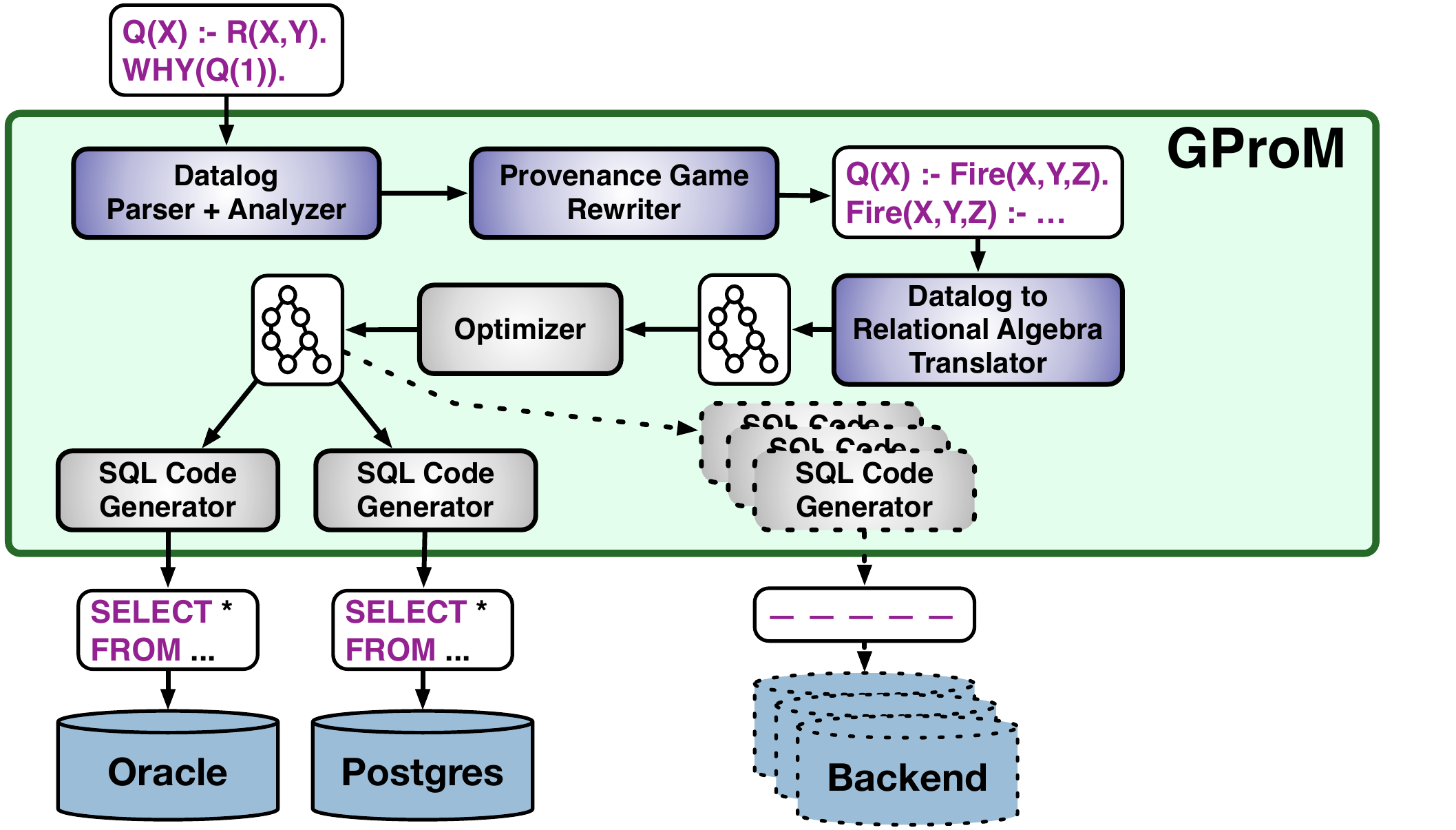}
  $\,$\\[-1mm]
  \caption{Implementation in GProM}
  \label{fig:gprom-gp}
\end{figure}

We have implemented the approach presented in Sec.~\ref{sec:compute-gp} in our provenance middleware called GProM~\cite{AG14} that executes provenance requests using a relational database backend (shown in Fig.\,\ref{fig:gprom-gp}). 
The system was originally developed to support provenance requests for SQL.
We have extended the system to support Datalog enriched with syntax for stating provenance questions. 
The user 
provides a why or why-not question and the corresponding Datalog query as input. 
Our system 
 parses  and semantically analyzes this input. Schema
information is gathered by querying the catalog of the backend database (e.g.,
to determine whether an EDB predicate with the expected arity exists). Modules
for accessing schema information are already part of the GProM system, but a new
semantic analysis component had to be developed to support Datalog. 
The algorithms presented in
Sec.~\ref{sec:compute-gp} are applied to create the program 
$\GPProg(P,Q(t),I)$
which computes 
$\explainq(P,Q(t),I)$.
This program is translated into a relational algebra ($\cal RA$) graph (GProM uses algebra
graphs instead of trees to allow for sharing of common subexpressions). The
algebra graph is then translated into SQL and sent 
to the backend database to compute the edge relation of 
 the explanation for the $\provQ$. Based on this edge relation, we then render a provenance graph 
(e.g., the graphs shown in 
Example\,\ref{ex:example1} are actual results produced by the system).\footnote{More examples for our 
  method and installation guideline for GProM are available at \url{https://github.com/IITDBGroup/gprom/wiki/datalog_prov}.}
While it would certainly be possible to directly translate the
Datalog program into SQL without the intermediate translation into 
$\cal RA$, we choose to introduce this step to be able to leverage the existing
heuristic and cost-based optimizations for $\cal RA$ 
graphs built into
GProM and use its library of 
$\cal RA$ to SQL translators. 

%
%
%
%
%
%
Our translation of first-order (FO) 
queries (a program with a distinguished answer relation) to $\cal RA$ 
is mostly standard. We first translate each rule into an algebra expression independently. 
Afterwards, we create expressions for IDB predicates as a union of the expressions for all rules with the predicates 
in the head. Finally, the algebra expressions for individual IDB predicates are connected into 
a single graph by replacing references to IDB predicates with their algebraic translation.  
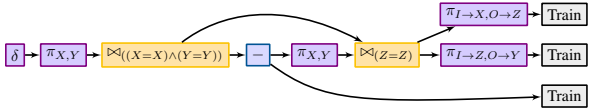
\begin{figure}[t]
  \centering
  $\,$\\[-3mm]
  \resizebox{0.9\linewidth}{!}{
   \begin{tikzpicture}[>=latex',line join=bevel,scale=0.5,line width=1pt]
  \pgfsetstrokecolor{black}
  \definecolor{fillcol}{rgb}{0,0,0};
  \pgfsetfillcolor{fillcol}
  \definecolor{strokecolor}{rgb}{0.5,0.0,0.5};
  \definecolor{fillcolor}{rgb}{0.84,0.8,1.0};
  \node (opNode0) at (20bp,158.0bp) [draw=strokecolor,fill=fillcolor,rectangle] {$\delta$};
  \definecolor{strokecolor}{rgb}{0.5,0.0,0.5};
  \definecolor{fillcolor}{rgb}{0.84,0.8,1.0};
  \node (opNode1) at (80bp,158.0bp) [draw=strokecolor,fill=fillcolor,rectangle] {$\pi_{X,Y}$};
  \definecolor{strokecolor}{rgb}{0.99,0.76,0.0};
  \definecolor{fillcolor}{rgb}{1.0,0.89,0.66};
  \node (opNode2) at (206.5bp,158.0bp) [draw=strokecolor,fill=fillcolor,rectangle] {$\bowtie_{((X = X)\wedge(Y = Y))}$};
  \definecolor{strokecolor}{rgb}{0.99,0.76,0.0};
  \definecolor{fillcolor}{rgb}{1.0,0.89,0.66};
  \node (opNode3) at (481.5bp,158.0bp) [draw=strokecolor,fill=fillcolor,rectangle] {$\bowtie_{(Z = Z)}$};
  \definecolor{strokecolor}{rgb}{0.5,0.0,0.5};
  \definecolor{fillcolor}{rgb}{0.84,0.8,1.0};
  \node (opNode4) at (601.5bp,208.5bp) [draw=strokecolor,fill=fillcolor,rectangle] {$\pi_{I \to X,O \to Z}$};
  \definecolor{strokecolor}{rgb}{0.0,0.0,0.0};
  \definecolor{fillcolor}{rgb}{0.93,0.93,0.93};
  \node (opNode5) at (701.5bp,208.5bp) [draw=strokecolor,fill=fillcolor,rectangle] {Train};
  \definecolor{strokecolor}{rgb}{0.5,0.0,0.5};
  \definecolor{fillcolor}{rgb}{0.84,0.8,1.0};
  \node (opNode6) at (601.5bp,158.0bp) [draw=strokecolor,fill=fillcolor,rectangle] {$\pi_{I \to Z,O \to Y}$};
  \definecolor{strokecolor}{rgb}{0.0,0.0,0.0};
  \definecolor{fillcolor}{rgb}{0.93,0.93,0.93};
  \node (opNode7) at (701.5bp,158.0bp) [draw=strokecolor,fill=fillcolor,rectangle] {Train};
  \definecolor{strokecolor}{rgb}{0.0,0.33,0.58};
  \definecolor{fillcolor}{rgb}{0.85,0.85,1.0};
  \node (opNode9) at (321.5bp,158.0bp) [draw=strokecolor,fill=fillcolor,rectangle] {$-$};
  \definecolor{strokecolor}{rgb}{0.5,0.0,0.5};
  \definecolor{fillcolor}{rgb}{0.84,0.8,1.0};
  \node (opNode10) at (391.5bp,158.0bp) [draw=strokecolor,fill=fillcolor,rectangle] {$\pi_{X,Y}$};
  \definecolor{strokecolor}{rgb}{0.0,0.0,0.0};
  \definecolor{fillcolor}{rgb}{0.93,0.93,0.93};
  \node (opNode14) at (701.5bp,107.0bp) [draw=strokecolor,fill=fillcolor,rectangle] {Train};

  \draw [->] (opNode0) -> (opNode1);
  \draw [->] (opNode1) -> (opNode2);
  \draw [->] (opNode6) -> (opNode7);
  \draw [->] (opNode2) -> (opNode9);
  \draw [->] (opNode4) -> (opNode5);
  \draw [->] (opNode10) -> (opNode3);
  \draw [->] (opNode9) -> (opNode10);
  \draw [->] (opNode2) ..controls (272.142bp,208.6bp) and (348.836bp,228.23bp)  .. (opNode3);
  \draw [->] (opNode3) -> (opNode6);
  \draw [->] (opNode9) ..controls (402.142bp,98.6bp) and (448.836bp,98.23bp)  .. (opNode14);
  \draw [->] (opNode3) -> (opNode4);

\end{tikzpicture}
  } 
  $\,$\\[-2mm]
  \caption{Translation for rule $r_1$}
  \label{fig:only2hop-translation}
\end{figure}

%
\begin{Example}\label{ex:example-translation}
Consider the translation of the rule $r_1$ from Fig.\,\ref{fig:running-example-db}. 
The $\cal RA$ 
graph for $r_1$ 
is shown in Fig.\,\ref{fig:only2hop-translation}. 
The translations of the first two goals are joined to compute the variable bindings 
for the 
positive part of the query. The negated goal is translated into a set difference 
between the positive part projected on $X,Y$ and relation $\rel{Train}$. 
The remaining 
three operators (from the left) join the positive with the negative part, 
project on the head variables, and remove duplicates.
\end{Example}



\BG{Removed the Optimization subsection unless we have something to say here}

\section{Experiments}
\label{sec:experiments}

We evaluate the performance of our solution 
over a co-author graph relation 
extracted from DBLP (\url{http://www.dblp.org/})
as well as over the TPC-H decision support benchmark (\url{http://www.tpc.org/tpch/default.asp}).
We compare our approach for computing explanations 
with the approach introduced for provenance games~\cite{KL13}. 
We call the provenance game approach \texttt{Direct Method(DM)}, because it 
directly constructs the full provenance graph. 
We have created subsets of the DBLP dataset with 100, 1K, 10K, 100K, 1M, 
and 8M 
co-author pairs (tuples). 
For the TPC-H benchmark, we used the following database sizes: 10MB, 100MB, 1GB, and 10GB.
All experiments were run on a machine with 2 x 3.3Ghz AMD Opteron 4238 CPUs (12 cores in total) and 128GB RAM running Oracle Linux 6.4. 
We use the commercial DBMS X (name omitted due to licensing restrictions) as a backend. 
Unless stated otherwise, each experiment was repeated 100 times and we report the median runtime.
We allocated a timeslot of 10 minutes for each run.
Computations that did not finish in the allocated time are omitted from the graphs.
\begin{figure}[t]
 \centering
  \begin{minipage}{0.98\linewidth}
    \centering
    $\,$\\[-3mm]
    \begin{minipage}{0.82\linewidth}
    \centering\scriptsize
    \begin{align*}
      r_1: \rel{only2hop}(X,Y) &\dlImp \rel{DBLP}(X,Z), \rel{DBLP}(Z,Y), \dlNeg \rel{DBLP}(X,Y)\\[1mm] \hline \\[-3mm]
      r_2: \rel{XwithYnotZ}(X,Y) &\dlImp \rel{DBLP}(X,Y), \dlNeg \rel{Q_1}(X)\\
      r_{2'}: \rel{Q_1}(X) &\dlImp \rel{DBLP}(X,\text{`Svein Johannessen'})\\[1mm] \hline \\[-3mm]
      r_3: \rel{only3hop}(X,Y) &\dlImp \rel{DBLP}(X,A), \rel{DBLP}(A,B), \rel{DBLP}(B,Y),\\ & \mathtab\mathtab \dlNeg \rel{E_1}(X), \dlNeg \rel{E_2}(X) \\
      r_{3'}: \rel{E_1}(X) &\dlImp \rel{DBLP}(X,Y)\\ 
      r_{3''}: \rel{E_2}(X) &\dlImp \rel{DBLP}(X,A), \rel{DBLP}(A,Y) 
      \\[1mm] \hline \\[-3mm]
      r_4: \rel{ordPriority}(X,Y) &\dlImp \rel{CUSTOMER}(A,X,B,C,D,E,F,G), \\ & \mathtab\mathtab \rel{ORDERS}(A,H,I,J,K,Y,M,N,O)\\[1mm] \hline 
    \end{align*}\\[-9mm]
    \begin{align*}
      r_5: \rel{ordDisc}(X,Y) &\dlImp \rel{CUSTOMER}(A,X,C,D,E,F,G,H), \\ 
	& \hspace{-11mm} \rel{ORDERS} (I,A,J,K,L,M,O,P,Q), \\  
	\hspace{-3mm} \rel{LINEITEM} & (I,R,S,T,U,V,Y,W,Z,A',B',C',D',E',F',G')\\[1mm] \hline \\[-3mm]
      r_6: \rel{partNotAsia}(X) &\dlImp \rel{PART}(A,X,B,C,D,E,F,G,H), \\ 
	&\rel{PARTSUPP} (A,I,J,K,L), \rel{SUPPLIER}(I,M,N,O,P,Q,R), \\  
	&\rel{NATION} (O,S,T,U), \dlNeg \rel{R_1}(T,\text{`ASIA'})\\ 
      r_{6'}: \rel{R_1}(T,Z) &\dlImp \rel{REGION}(T,Z,V)\\
    \end{align*}
    \end{minipage}
  \end{minipage}
  $\,$\\[-4mm]
  \caption{DBLP and TPC-H queries for experiments}
  \label{fig:experi-queries}
\end{figure}

\SL{$r_6$ still remains to be decided.}
\mypartitle{Workloads}
We compute explanations for the queries in Fig.\,\ref{fig:experi-queries} over 
the datasets we have introduced. 
For the DBLP dataset, we consider: \rel{only2hop} ($r_1$) which is our running example query; 
\rel{XwithYnotZ} ($r_2$) that returns authors that are direct co-authors of a certain person $Y$, 
but not of ``Svein Johannessen''; 
\rel{only3hop} ($r_3$) that returns pairs of authors $(X,Y)$ that are connected via a path of length 3 in the co-author graph where $X$ is not a co-author or indirect co-author (2 hops) of $Y$. 
For TPC-H, we consider: \rel{ordPriority} ($r_4$) which returns for each customer the priorities of her/his orders;
\rel{ordDisc} ($r_5$) which returns customers and the discount rates of items in their orders; 
\rel{partNotAsia} ($r_6$) which finds parts that can be supplied from a country that is not in Asia. 
\mypartitle{Implementing DM}
As introduced in Sec.~\ref{sec:intro}, 
\texttt{DM} has to instantiate a 
graph with ${\cal O}(\card{\adom{I}}^n)$ nodes
where $n$ is the maximal number of variables in a rule. 
%
%
We do not have a full implementation of \texttt{DM}, 
but can compute a conservative 
lower bound for the runtime of the step constructing the game graph 
by executing  
a query that computes an $n$-way cross-product over the active domain.
Note that the actual runtime will be much higher 
because 1) several edges are created for each rule binding (we underestimate the number of nodes of the constructed graph) and 
2) recursive Datalog queries have to be evaluated over this graph using the well-founded semantics. 
The results for different instance sizes and number of variables are shown in 
Fig.\,\ref{tab:baseline}.
Even for only 2 variables, \texttt{DM} did not finish
for datasets of more than 10K tuples
within the allocated 10 min timeslot. 
For queries with more than 4 variables, 
\texttt{DM} did not even finish for the smallest dataset. 
%
\begin{figure}
  \begin{center}
	\scriptsize \centering $\,$\\[-2.5mm]
    \begin{tabular}{|l|c|c|c|c|} \hline
	\thead {Num of Vars \textbackslash\, DBLP (\#tuples)} & \thead {100} & \thead {1K} & \thead {10K} & \thead {100K}
		\\ \hline
		2 Variables ($r_2$) & 0.043 & 0.171 & 14.016 & -
    		\\ \hline
	        3 Variables ($r_1$) & 0.294 & 285.524 & - & -
    		\\ \hline
	        4 Variables ($r_3$) & 56.070 & - & - & -
		\\ \hline
	\thead {Num of Vars \textbackslash\, TPC-H (Size)} & \thead {10MB} & \thead {100MB} & \thead {1GB} & \thead {10GB}
    		\\ \hline
		($>$ 10) Variables ($r_4$, $r_5$, $r_6$) & - & - & - & -
    		\\ \hline
    \end{tabular}
  \end{center}
  $\,$\\[-6mm]
  \caption{Runtime of \texttt{DM} in seconds. For entries with `-', the computation did not finish in the allocated time of 10 min.}
  \label{tab:baseline}
\end{figure}


\begin{figure}[t]
\begin{minipage}{0.98\linewidth}
$\,$\\[-7mm]
\scriptsize\centering
\subfloat[\small Runtime of \rel{only2hop}]
{
  \includegraphics[width=0.50\columnwidth,trim=0 60 0 0, clip]{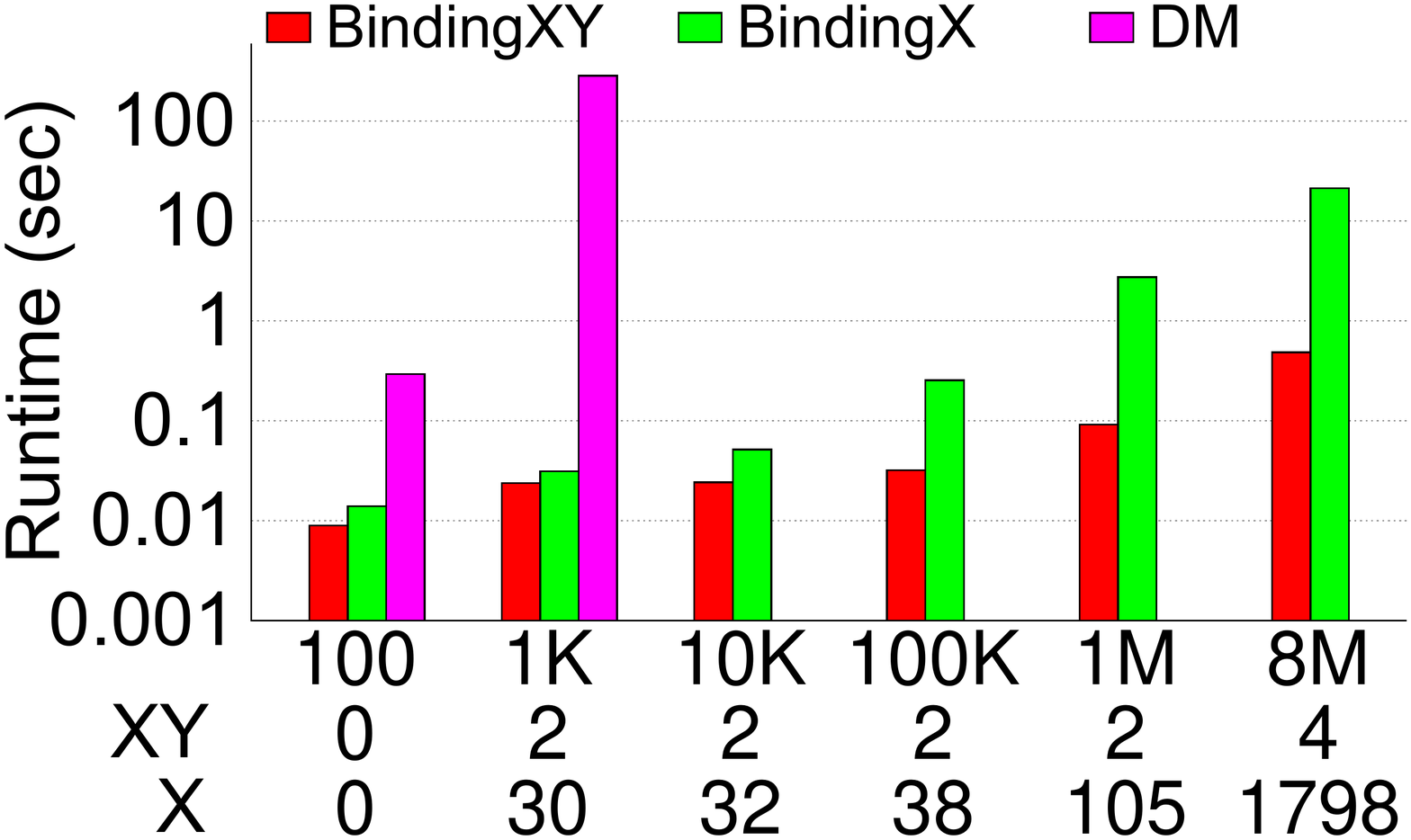}
}
\subfloat[\small Runtime of \rel{XwithYnotZ}]
{
  \includegraphics[width=0.50\columnwidth,trim=0 60 0 0, clip]{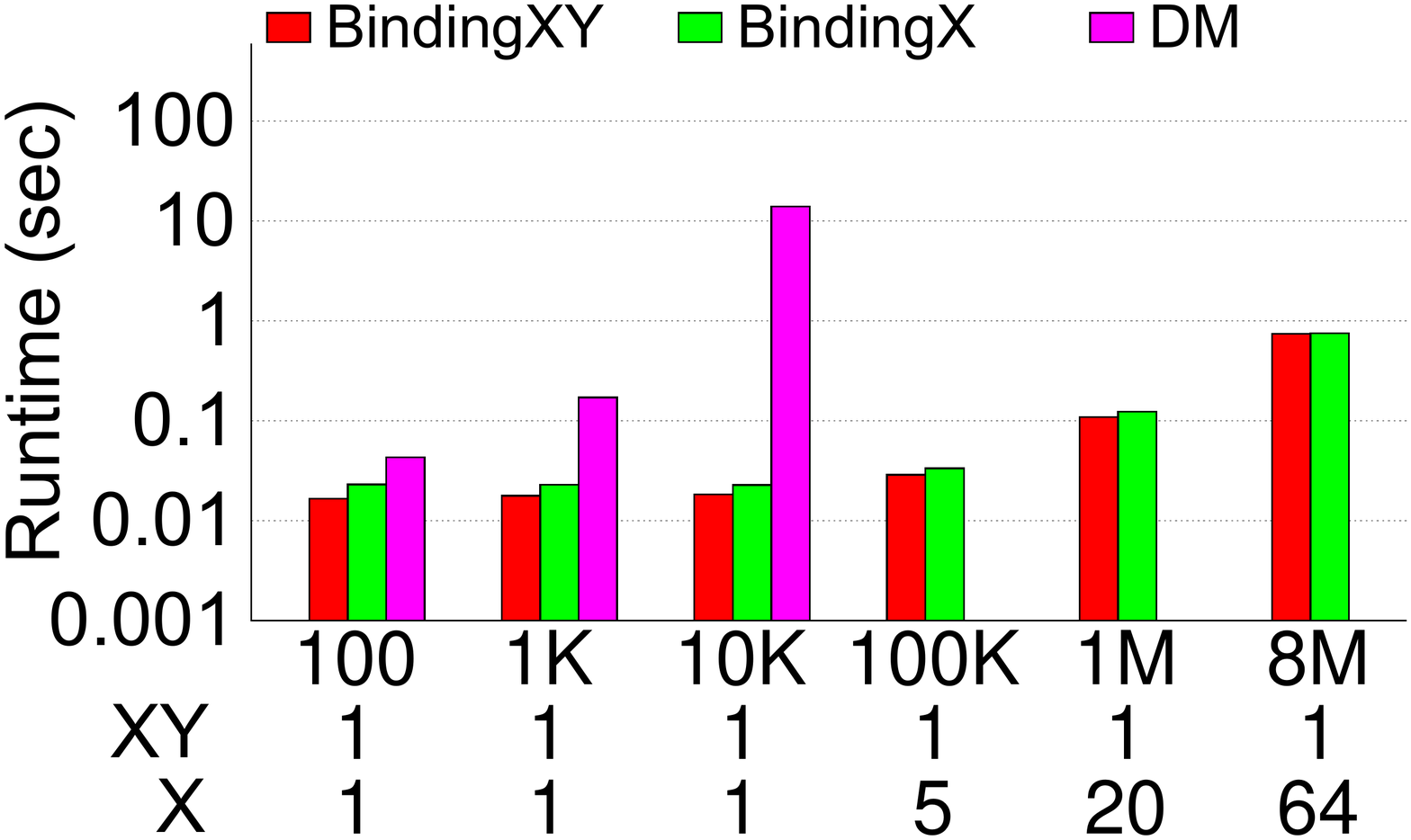}
}
\\[2mm]
\begin{minipage}{0.98\linewidth}
\scriptsize\centering
\subfloat[\small Variable bindings for DBLP $\provQ$s]{
  \begin{minipage}{0.98\linewidth}
  \centering$\,$\\[-1mm]
  \begin{tabular}{|c|cc|}
    \thead{Query \textbackslash\, Binding}&\thead{X}&\thead{Y}\\
    (a) $\rel{only2hop}$ & Tore Risch & Rafi Ahmed \\
    (b) \rel{XwithYnotZ} & Arjan Durresi & Raj Jain \\
    \hline
  \end{tabular}
\end{minipage}
}
\end{minipage}
\\[-6mm]
\subfloat[\small Runtime of \rel{ordPriority}]{
  \includegraphics[width=0.50\columnwidth,trim=0 60 0 0, clip]{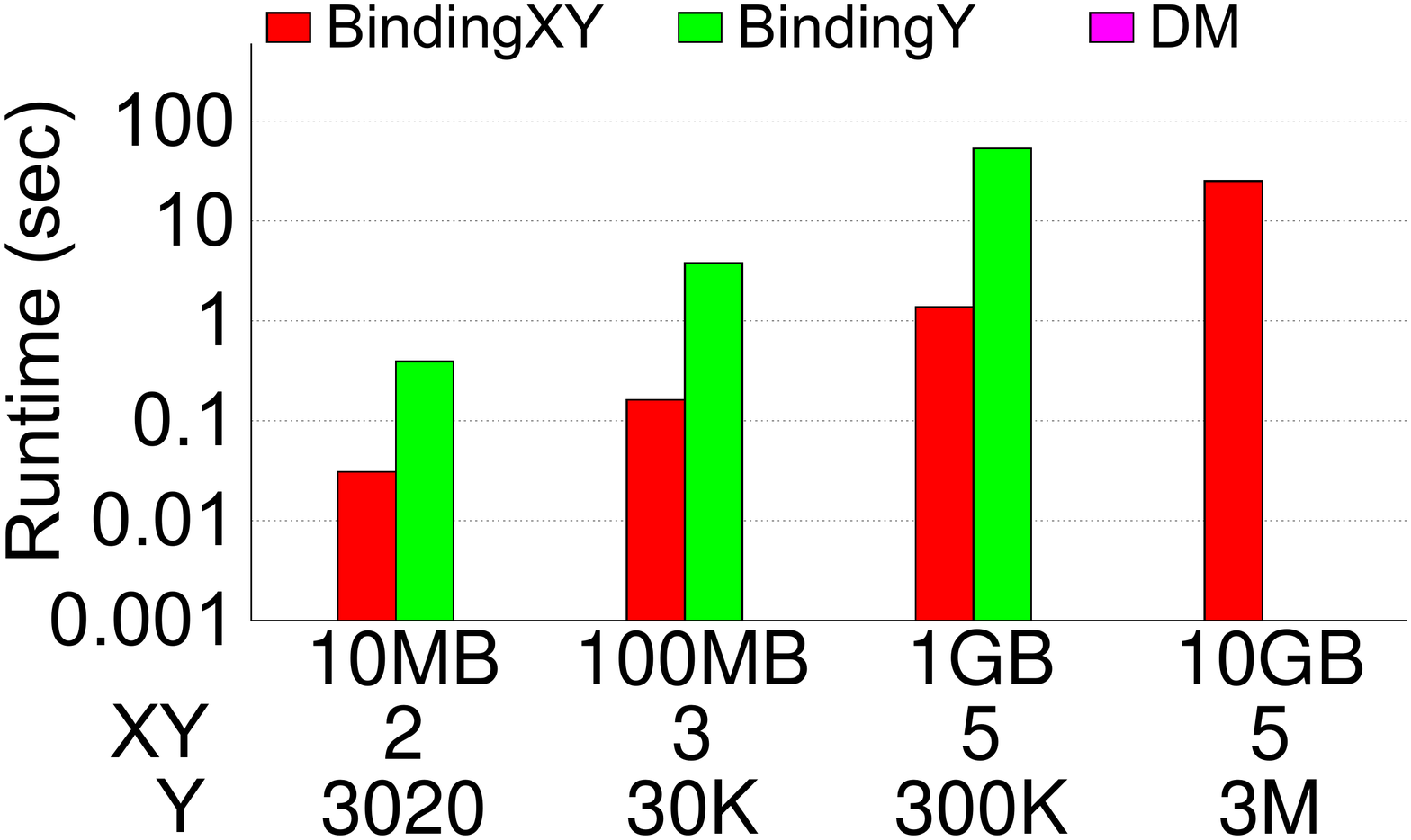}
}
\subfloat[\small Runtime of \rel{ordDisc}]{
  \includegraphics[width=0.50\columnwidth,trim=0 60 0 0, clip]{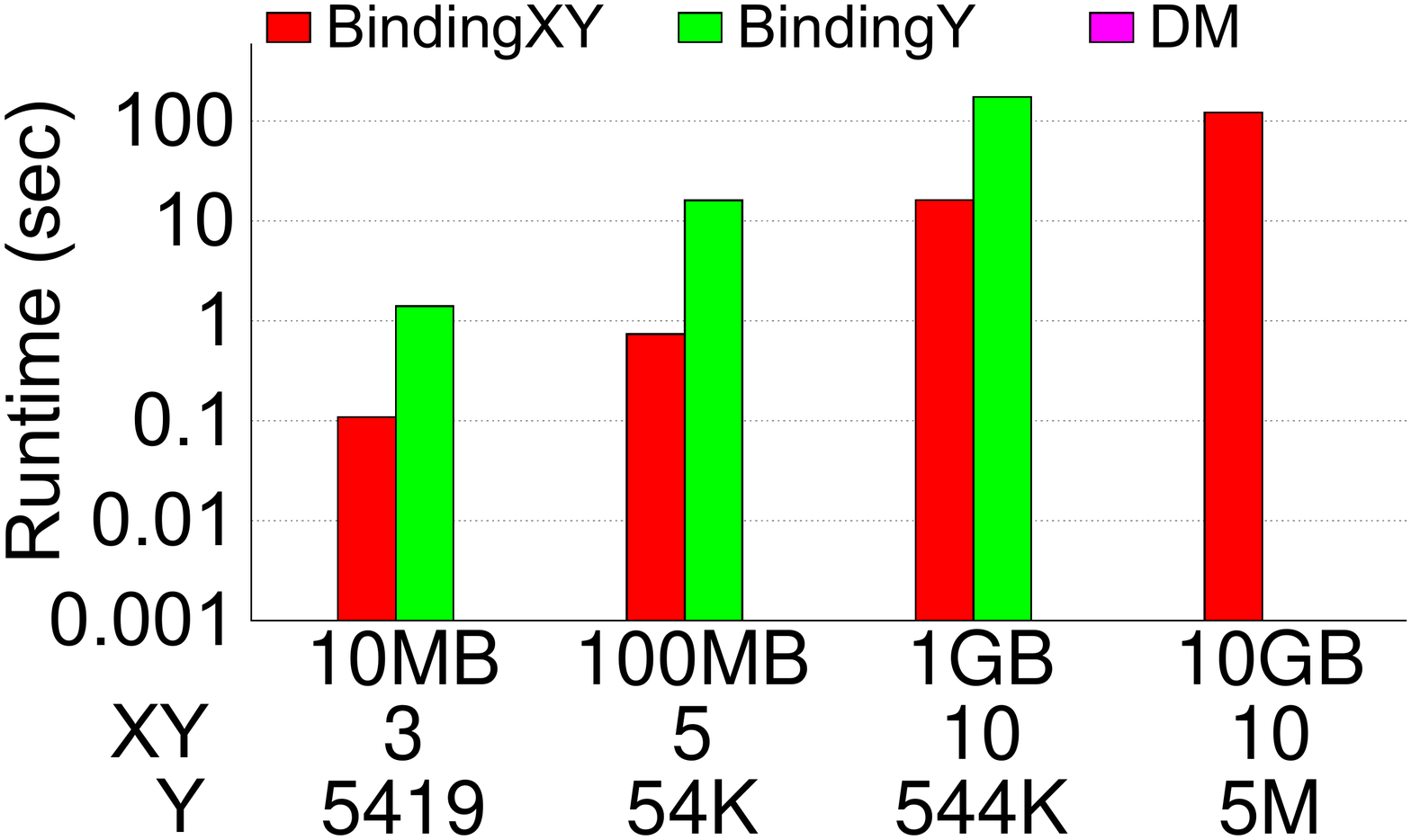}
}
\end{minipage}
\\[3mm]
\begin{minipage}{0.98\linewidth}
\scriptsize\centering
\subfloat[\small Variable bindings for TPC-H $\provQ$s]{
  \begin{minipage}{0.98\linewidth}
 \centering$\,$\\[-1mm]
\begin{tabular}{|c|cc|}
\thead{Query \textbackslash\, Binding}&\thead{X}&\thead{Y}\\
(d) $\rel{ordPriority}$ & Customer16 & 1-URGENT \\
(e) $\rel{ordDisc}$ & Customer16 & 0 \\
\hline
\end{tabular}
\end{minipage}
}
\end{minipage}
$\,$\\[-2mm]
\caption{Runtime - Why questions}
\label{fig:perf-why}
\end{figure}

\mypartitle{Why Questions}
The runtime incurred for generating explanations for 
 why questions 
over the queries $r_1$, $r_2$, $r_4$, and $r_5$ (Fig.\,\ref{fig:experi-queries}) is shown in Fig.\,\ref{fig:perf-why}.
For the evaluation, 
we consider the effect of 
the different binding patterns on performance.
Fig.\,\ref{fig:perf-why}.c and \ref{fig:perf-why}.f show the variables bound by the PQs we have considered. 
%
Fig.\,\ref{fig:perf-why}.a and \ref{fig:perf-why}.b show 
the performance results for $r_1$ and $r_2$, respectively. 
We also show 
number of rule nodes in the provenance graph
for each binding pattern below the X axis.
If only variable $X$ 
is bound (\texttt{BindingX}), 
then the queries determine 
authors 
that  
occur together with the author we have bound to $X$
in the query result. 
For instance, the explanation derived 
for 
$\rel{only2hop}$ 
with 
\texttt{BindingX}
(Fig.\,\ref{fig:perf-why}.a) 
explains why persons are indirect, but not direct, co-authors of ``Tore Risch''. 
If both 
$X$ and $Y$ are bound (\texttt{BindingXY}), 
then the explanation for $r_1$ and $r_2$ is limited to a particular indirect and direct 
co-author, respectively.
The runtime for generating explanations using our approach 
exhibits roughly linear growth in the dataset size
and dominates 
\texttt{DM} even for the small instances.
Furthermore, Fig.\,\ref{fig:perf-why}.d and \ref{fig:perf-why}.e 
(for $r_4$ and $r_5$, respectively) show that our approach can 
handle queries with many variables where \texttt{DM} times out 
even for 
the smallest dataset we have considered.
%
Binding one variable, e.g., \texttt{BindingY}, in queries $r_4$ and $r_5$ 
expresses 
a 
condition, 
e.g., $Y$ = `1-URGENT' in $r_4$ requires the order priority to be urgent.
If both variables are bound, 
then the provenance question verifies the existence of orders for a certain customer 
(e.g., why ``Customer16'' has at least one urgent order). 
Runtimes exhibit the same trend as for the DBLP queries. 
\begin{figure}[t]
\begin{minipage}{0.98\linewidth}
$\,$\\[-11mm]
\centering
\subfloat[\small Runtime of \rel{only2hop}]{\includegraphics[width=0.50\columnwidth,trim=0 60 0 0, clip]{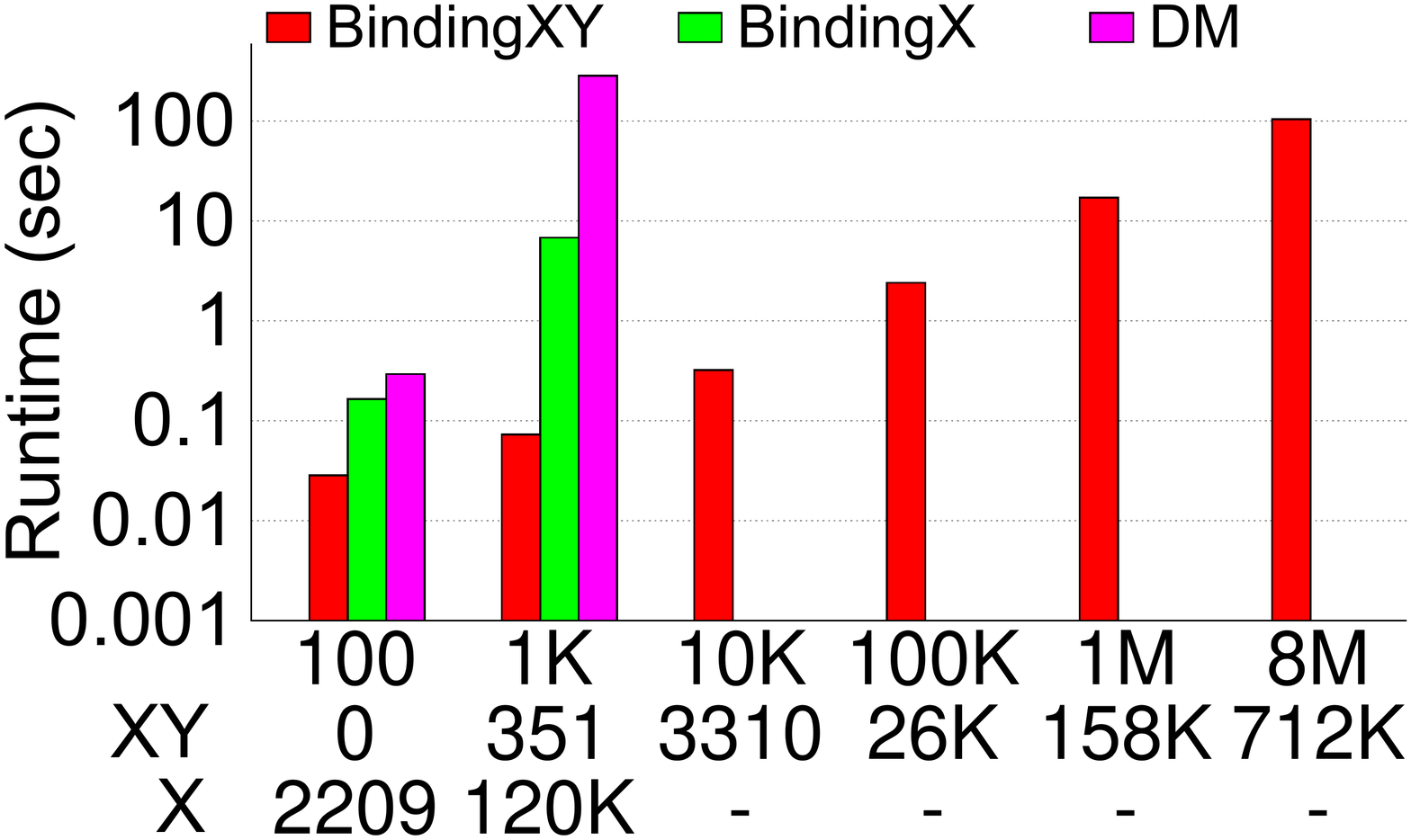}}
\subfloat[\small Runtime of \rel{XwithYnotZ}]{\includegraphics[width=0.50\columnwidth,trim=0 60 0 0, clip]{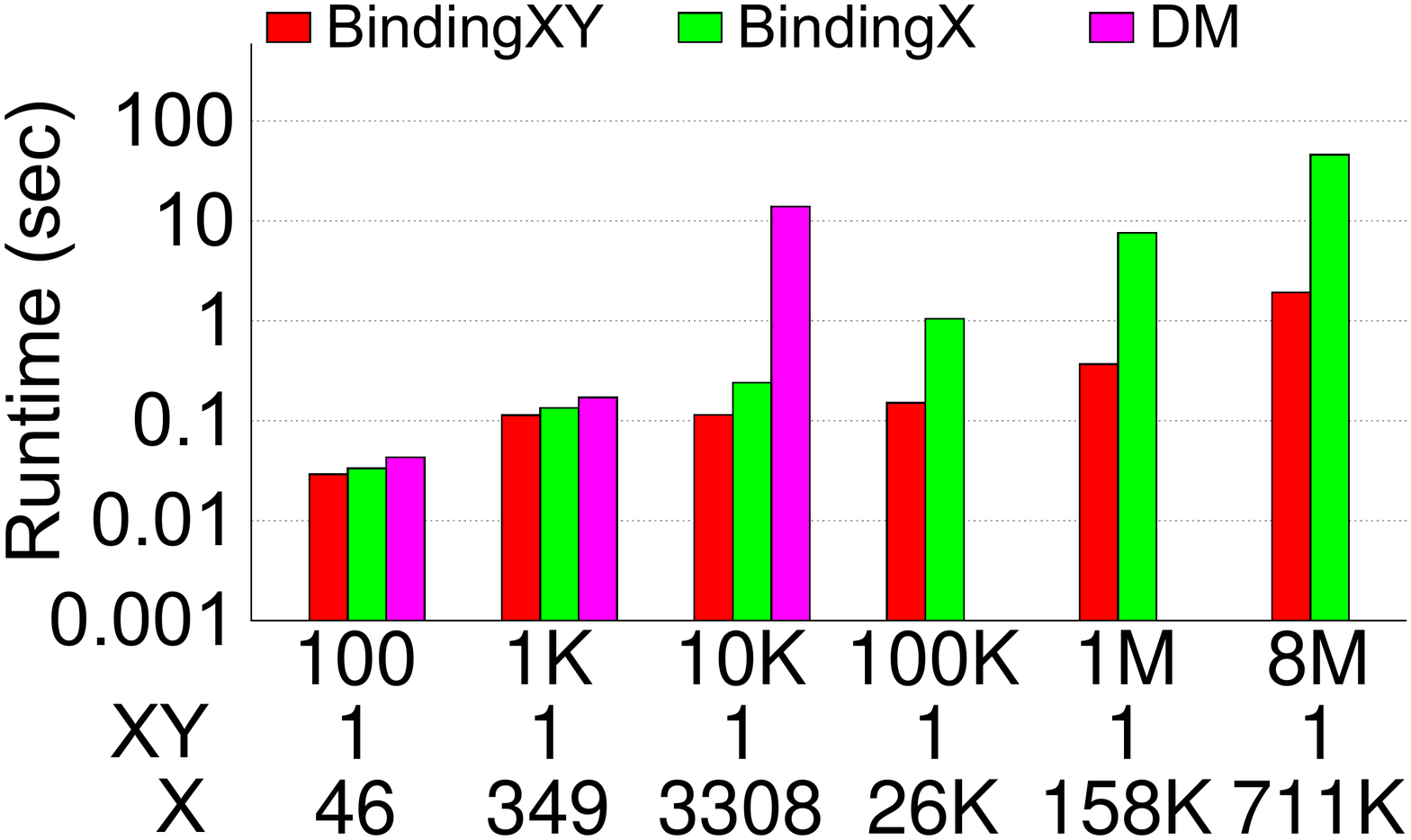}}
\end{minipage}
\\[3mm]
\begin{minipage}{0.98\linewidth}
\scriptsize\centering$\,$\\[-1mm]
\subfloat[\small Variable bindings for DBLP $\provQ$s]{
\begin{tabular}{|c|cc|}
\thead{Query \textbackslash\, Binding}&\thead{X}&\thead{Y}\\
(a) \rel{only2hop} & Tore Risch & Svein Johannessen \\
(b) \rel{XwithYnotZ} & Tor Skeie & Joo-Ho Lee\\
\hline
\end{tabular}
}
\end{minipage}
$\,$\\[-2mm]
\caption{Runtime - Why-not questions}
\label{fig:perf-whynot}
\end{figure}

\mypartitle{Why-not Provenance}
We 
have queries $r_1$ and $r_2$ from Fig.\,\ref{fig:experi-queries}
to evaluate the performance of computing explanations for failed derivations.
When binding all variables in the $\provQ$ (\texttt{BindingXY}) with the information in Fig.\,\ref{fig:perf-whynot}.c, 
these queries check if a particular set of authors cannot appear together in the result.
For instance, for $\rel{only2hop}$ $(r_1)$ the query checks why ``Tore Risch'' is either
not an indirect co-author 
or 
a direct co-author of ``Svein Johannessen''.
If one variable is bound (\texttt{BindingX}), then 
the why-not question explains for pairs of authors where one of the authors is bound to $X$, why the pair does not appear together in the query result.
The results for queries on $r_1$ and $r_2$ are shown in 
Fig.\,\ref{fig:perf-whynot}.a and \ref{fig:perf-whynot}.b., respectively.
The number of output tuples produced by the provenance computation 
(based on the number of rule nodes shown below the X axis) 
is quadratic in the database size resulting in a quadratice increase in runtime. 
Our approach improves the performance over large instances in comparison to \texttt{DM},
which is limited to very small datasets (less than 10K). 
%
%
Limiting the result size of missing answer questions 
for queries with many existential variables ($r_4$, $r_5$, and $r_6$) 
would 
require 
aggressive summarization 
techniques, which we will address 
in future work.

\begin{figure}[t]
\begin{minipage}{0.98\linewidth}
$\,$\\[-7mm]
\centering
\subfloat[\small Runtime of \rel{only3hop}]{\includegraphics[width=0.50\columnwidth,trim=0 60 0 0, clip]{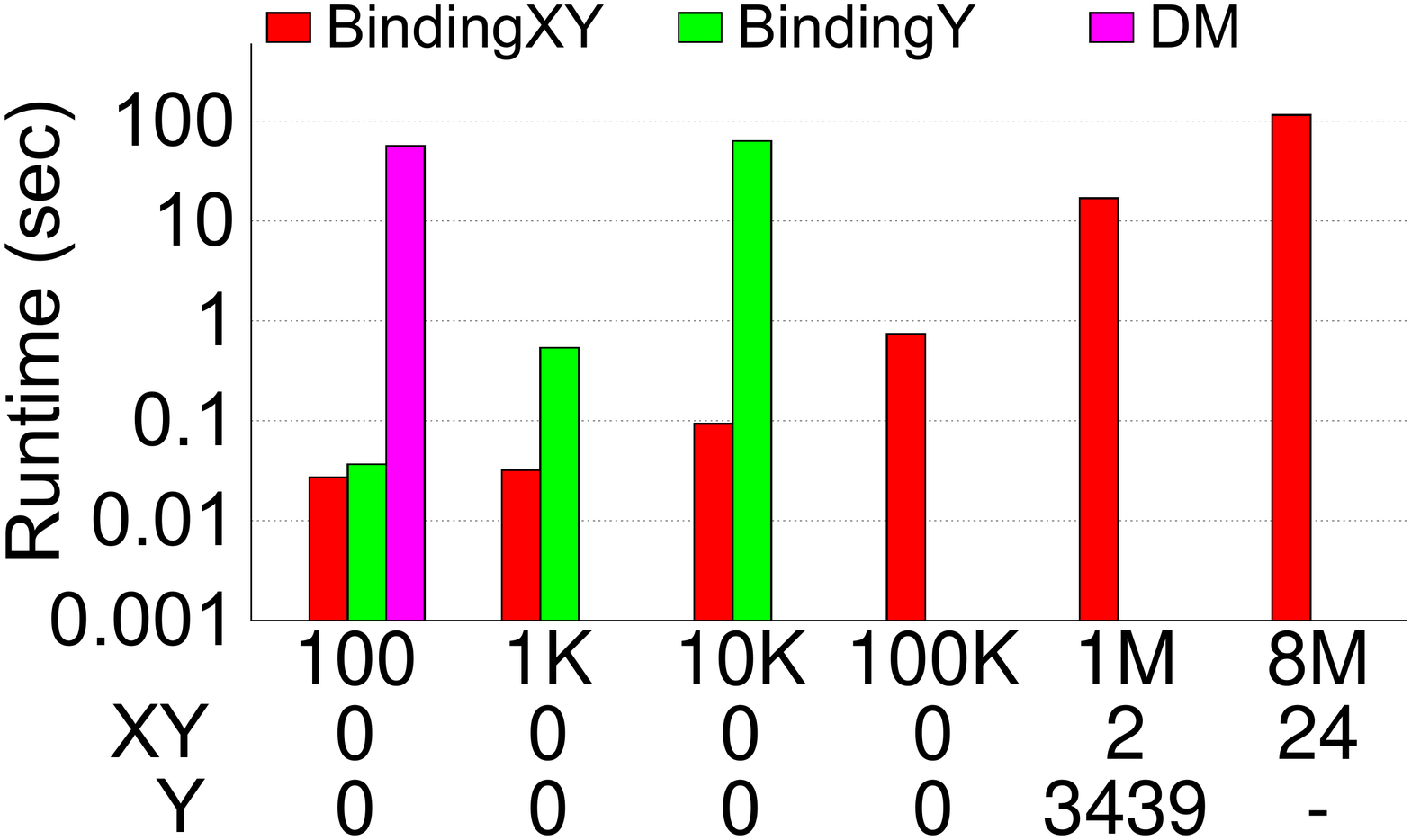}}
\subfloat[\small Runtime of \rel{partNotAsia}]{\includegraphics[width=0.50\columnwidth,trim=0 60 0 0, clip]{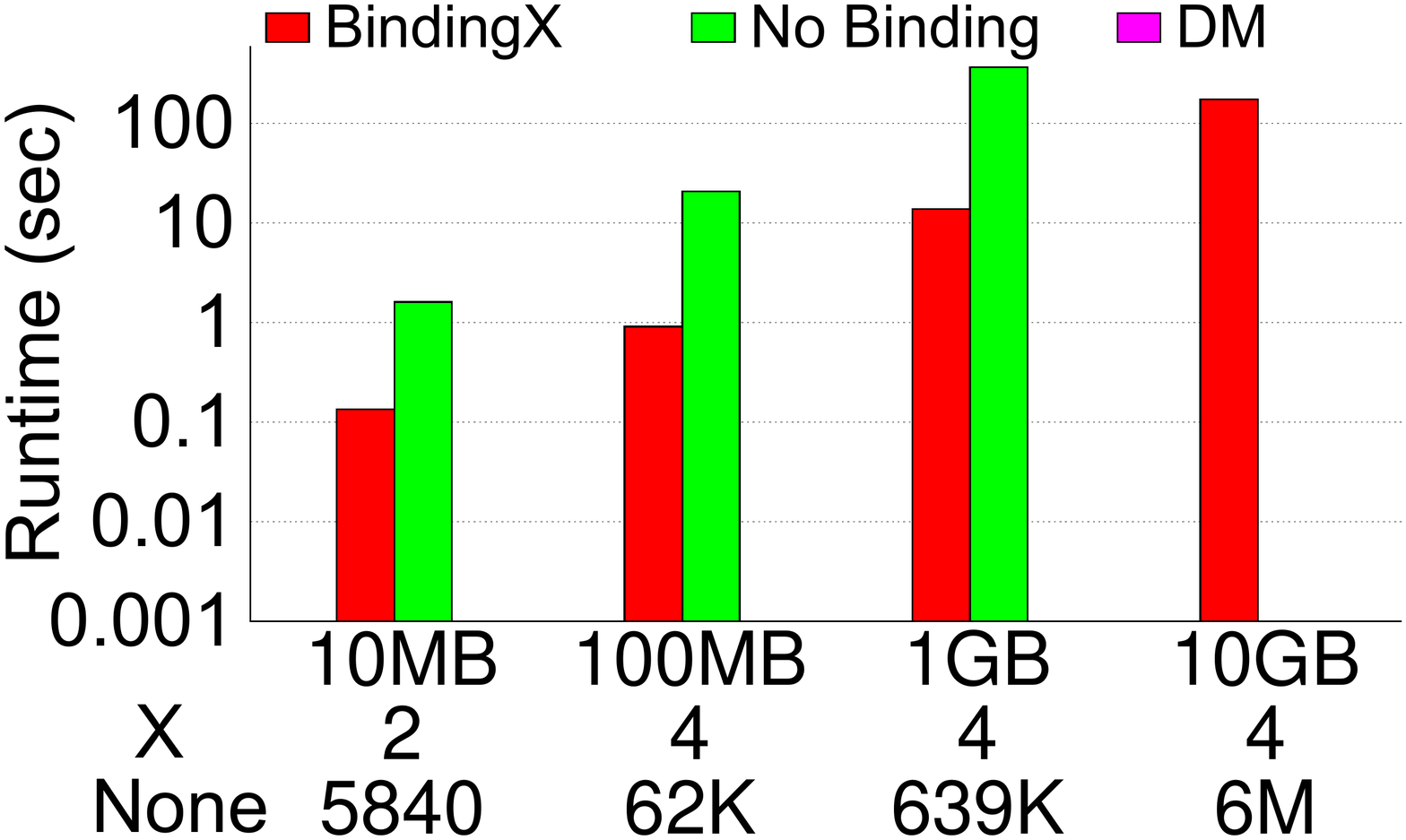}}
\end{minipage}
\\[3mm]
\begin{minipage}{0.98\linewidth}
\scriptsize\centering
\subfloat[\small Variable bindings for DBLP and TPC-H $\provQ$s]{
  \begin{minipage}{0.98\linewidth}
  \centering$\,$\\[-1mm]
\begin{tabular}{|c|cc|}
\thead{Query \textbackslash\, Binding}&\thead{X} &\thead{Y}\\
(a) $\rel{only3hop}$ & Alex Benton & Paul Erdoes \\
(b) $\rel{partNotAsia}$ & grcpi$^{1}$ & - \\
\hline
\end{tabular}\\[1mm]
${}^{1}$~grcpi = ghost royal chocolate peach ivory
\end{minipage}
}
\end{minipage}
$\,$\\[-2mm]
\caption{Runtime - Why questions over queries with negation}
\label{fig:perf-uni}
\end{figure}
\SL{The graph for TPCH above should be updated.}
\mypartitle{Queries with Negation}
Recall that our approach 
 handles queries with negation.
We choose 
rules $r_3$ (multiple negated goals)
and $r_6$ (one negated goal) shown in Fig.\,\ref{fig:experi-queries} to evaluate the performance of asking why questions over such queries. 
We use the bindings shown in Fig.\,\ref{fig:perf-uni}.c.
The results shown in Fig.\,\ref{fig:perf-uni}.a and \ref{fig:perf-uni}.b demonstrate that 
our approach 
efficiently computes explanations for $r_3$ and $r_6$, respectively.
When increasing the database size, the runtimes of PQs for these queries exhibit the same trend 
as observed for other why (why-not) questions and significantly outperform \texttt{DM}.
For instance, the performance of the query 
$\rel{partNotAsia}$ (Fig.\,\ref{fig:perf-uni}.b),
which contains many variables and negation  
exhibits
the same trend as 
 queries that have no negation 
(i.e., $r_4$ and $r_5$). 

\section{Conclusions}
\label{sec:concl}

We present a unified framework for explaining answers and non-answers over first-order (FO) queries. 
Our approach is based on the concept of firing rules that we extend to support negation and missing answers. 
Our 
efficient middleware implementation generates a Datalog program that 
computes the explanation 
for a 
provenance question and compiles this program 
into SQL. 
Our experimental evaluation demonstrates that  
by avoiding to generate irrelevant parts of the graph 
for the provenance question 
we can answer provenance questions over large instances.
An interesting avenue for future work is to investigate summarized 
representation of provenance (e.g., in the spirit of ~\cite{CC14,EA14,GK15,RK14}) 
to deal with the large size of explanations for missing answers. 
Other topics of interest include considering integrity constraints in the provenance graph construction 
(e.g., rule derivations 
can never succeed if they violate integrity constrains), 
marrying the approach with ideas from missing answer approaches that only return one explanation 
that is optimal according to some criterion, and extending the approach for more expressive query languages 
(e.g., aggregation or non-stratified recursive programs).


\bibliographystyle{abbrv}
\bibliography{cp}

\begin{thebibliography}{10}

\bibitem{AG14}
B.~Arab, D.~Gawlick, V.~Radhakrishnan, H.~Guo, and B.~Glavic.
\newblock A generic provenance middleware for database queries, updates, and
  transactions.
\newblock In {\em TaPP}, 2014.

\bibitem{BMSU86}
F.~Bancilhon, D.~Maier, Y.~Sagiv, and J.~D. Ullman.
\newblock Magic sets and other strange ways to implement logic programs.
\newblock In {\em PODS}, pages 1--15, 1986.

\bibitem{BH14a}
N.~Bidoit, M.~Herschel, and K.~Tzompanaki.
\newblock Immutably answering why-not questions for equivalent conjunctive
  queries.
\newblock In {\em TaPP}, 2014.

\bibitem{BH14}
N.~Bidoit, M.~Herschel, K.~Tzompanaki, et~al.
\newblock {Query-Based Why-Not Provenance with NedExplain}.
\newblock In {\em EDBT}, pages 145--156, 2014.

\bibitem{CC14}
B.~t. Cate, C.~Civili, E.~Sherkhonov, and W.-C. Tan.
\newblock High-level why-not explanations using ontologies.
\newblock In {\em PODS}, pages 31--43, 2014.

\bibitem{CJ09}
A.~Chapman and H.~V. Jagadish.
\newblock {Why Not?}
\newblock In {\em SIGMOD}, pages 523--534, 2009.

\bibitem{cheney2009provenance}
J.~Cheney, L.~Chiticariu, and W.~Tan.
\newblock Provenance in databases: Why, how, and where.
\newblock {\em Foundations and Trends in Databases}, 1(4):379--474, 2009.

\bibitem{DG15c}
D.~Deutch, A.~Gilad, and Y.~Moskovitch.
\newblock Selective provenance for datalog programs using top-k queries.
\newblock {\em PVLDB}, 8(12):1394--1405, 2015.

\bibitem{DM14c}
D.~Deutch, T.~Milo, S.~Roy, and V.~Tannen.
\newblock Circuits for datalog provenance.
\newblock In {\em ICDT}, pages 201--212, 2014.

\bibitem{EA14}
K.~El~Gebaly, P.~Agrawal, L.~Golab, F.~Korn, and D.~Srivastava.
\newblock Interpretable and informative explanations of outcomes.
\newblock {\em PVLDB}, 8(1):61--72, 2014.

\bibitem{FK97}
J.~Flum, M.~Kubierschky, and B.~Lud{\"a}scher.
\newblock Total and partial well-founded datalog coincide.
\newblock In {\em ICDT}, pages 113--124, 1997.

\bibitem{GK15}
B.~Glavic, S.~K\"{o}hler, S.~Riddle, and B.~Lud\"{a}scher.
\newblock Towards constraint-based explanations for answers and non-answers.
\newblock In {\em TaPP}, 2015.

\bibitem{GM13}
B.~Glavic, R.~J. Miller, and G.~Alonso.
\newblock Using sql for efficient generation and querying of provenance
  information.
\newblock In {\em In search of elegance in the theory and practice of
  computation}, pages 291--320. 2013.

\bibitem{GA12}
T.~J. Green, M.~Aref, and G.~Karvounarakis.
\newblock Logicblox, platform and language: A tutorial.
\newblock In {\em Datalog in Academia and Industry}, pages 1--8. Springer,
  2012.

\bibitem{GK07a}
T.~J. Green, G.~Karvounarakis, Z.~G. Ives, and V.~Tannen.
\newblock {Update Exchange with Mappings and Provenance}.
\newblock In {\em VLDB}, pages 675--686, 2007.

\bibitem{HH10}
M.~Herschel and M.~Hernandez.
\newblock {Explaining Missing Answers to SPJUA Queries}.
\newblock {\em PVLDB}, 3(1):185--196, 2010.

\bibitem{huang2008provenance}
J.~Huang, T.~Chen, A.~Doan, and J.~Naughton.
\newblock On the provenance of non-answers to queries over extracted data.
\newblock In {\em VLDB}, pages 736--747, 2008.

\bibitem{grigoris-tj-simgodrec-2012}
G.~Karvounarakis and T.~J. Green.
\newblock Semiring-annotated data: queries and provenance.
\newblock {\em SIGMOD Record}, 41(3):5--14, 2012.

\bibitem{kohler2012declarative}
S.~K{\"o}hler, B.~Lud{\"a}scher, and Y.~Smaragdakis.
\newblock Declarative datalog debugging for mere mortals.
\newblock In {\em Datalog 2.0: Datalog in Academia and Industry}, pages
  111--122, 2012.

\bibitem{KL13}
S.~K{\"o}hler, B.~Lud{\"a}scher, and D.~Zinn.
\newblock First-order provenance games.
\newblock In {\em In Search of Elegance in the Theory and Practice of
  Computation}, pages 382--399. 2013.

\bibitem{LS16}
S.~Lee, S.~K\"{o}hler, B.~Lud\"{a}scher, and B.~Glavic.
\newblock {Implementing Unified Why- and Why-Not Provenance Through Games}.
\newblock In {\em TaPP (Poster)}, 2016.

\bibitem{MG10}
A.~Meliou, W.~Gatterbauer, K.~Moore, and D.~Suciu.
\newblock {The Complexity of Causality and Responsibility for Query Answers and
  non-Answers}.
\newblock {\em PVLDB}, 4(1):34--45, 2010.

\bibitem{PS09}
E.~Pontelli, T.~C. Son, and O.~Elkhatib.
\newblock Justifications for logic programs under answer set semantics.
\newblock {\em Theory and Practice of Logic Programming}, 9(01):1--56, 2009.

\bibitem{RK14}
S.~Riddle, S.~K\"ohler, and B.~Lud\"ascher.
\newblock Towards constraint provenance games.
\newblock In {\em TaPP}, 2014.

\bibitem{TC10}
Q.~T. Tran and C.-Y. Chan.
\newblock How to conquer why-not questions.
\newblock In {\em SIGMOD}, pages 15--26, 2010.

\bibitem{WZ14a}
Y.~Wu, M.~Zhao, A.~Haeberlen, W.~Zhou, and B.~T. Loo.
\newblock Diagnosing missing events in distributed systems with negative
  provenance.
\newblock In {\em SIGCOMM}, pages 383--394, 2014.

\bibitem{ZS10}
W.~Zhou, M.~Sherr, T.~Tao, X.~Li, B.~T. Loo, and Y.~Mao.
\newblock Efficient querying and maintenance of network provenance at
  internet-scale.
\newblock In {\em SIGMOD}, pages 615--626, 2010.

\end{thebibliography}

\end{document}